\documentclass[12pt]{article} 
\usepackage[sectionbib]{natbib}
\usepackage{array,epsfig,fancyheadings,rotating}
\usepackage[]{hyperref}  
\usepackage{caption, setspace}
\usepackage{sectsty, secdot}
\sectionfont{\fontsize{12}{14pt plus.8pt minus .6pt}\selectfont}
\renewcommand{\theequation}{\thesection\arabic{equation}}
\subsectionfont{\fontsize{12}{14pt plus.8pt minus .6pt}\selectfont}

\usepackage{amsmath}
\usepackage{amssymb}
\usepackage{amsfonts}
\usepackage{amsthm}

\newtheorem{theorem}{Theorem}
\newtheorem{condition}{Condition}

\newtheorem{corollary}{Corollary}

\theoremstyle{definition}

\usepackage{graphicx,psfrag,epsf}
\usepackage{enumerate}
\usepackage{natbib}
\usepackage{url} %
\usepackage{booktabs,threeparttable}

\usepackage{moreverb}
\usepackage{adjustbox}
\usepackage{xcolor}
\usepackage{multirow}

\usepackage{times}
\usepackage{bm}

\def\bSig\mathbf{\Sigma}

\newcommand\red[1]{{\color{black}#1}}
\newcommand\blue[1]{{\color{black}#1}}
\newcommand\redreview[1]{{\color{black}#1}}
\newcommand\chixiang[1]{{\color{black}#1}}

\def\bfn{{\ensuremath{\bf n}}}

\def\bfg{{\ensuremath{\bf g}}}

\def\bfD{{\ensuremath{\bf D}}}

\def\bfQ{{\ensuremath{\bf Q}}}

\def\bfU{{\ensuremath{\bf U}}}
\def\bfV{{\ensuremath{\bf V}}}

\def\bfX{{\ensuremath{\bf X}}}
\def\bfY{{\ensuremath{\bf Y}}}

\def\bfdelta{{\ensuremath\boldsymbol{\delta}}}

\def\bfbeta{{\ensuremath\boldsymbol{\beta}}}
\def\bfgamma{{\ensuremath\boldsymbol{\gamma}}}

\def\bfOmega{{\ensuremath\boldsymbol{\Omega}}}

\def\bfSigma{{\ensuremath\boldsymbol{\Sigma}}}

\def\bfmu{{\ensuremath{{\boldsymbol{\mu}}}}}

\def\bfg{{\ensuremath{\bf g}}}

\def\bfQ{{\ensuremath{\bf Q}}}

\def\bfU{{\ensuremath{\bf U}}}
\def\bfV{{\ensuremath{\bf V}}}

\def\bfX{{\ensuremath{\bf X}}}
\def\bfY{{\ensuremath{\bf Y}}}

\def\bfI{{\ensuremath{\bf I}}}
\def\bfH{{\ensuremath{\bf H}}}

\def\bfbeta{{\ensuremath\boldsymbol{\beta}}}

\def\bfmu{{\ensuremath\boldsymbol{\mu}}}

\def\bfOmega{{\ensuremath\boldsymbol{\Omega}}}

\def\bfSigma{{\ensuremath\boldsymbol{\Sigma}}}

\def\bfepsilon{{\ensuremath\boldsymbol{\epsilon}}}

\def\bflambda{{\ensuremath\boldsymbol{\lambda}}}

\def\bfrho{{\ensuremath\boldsymbol{\rho}}}

\def\T{{ \mathrm{\scriptscriptstyle T} }}

\begin{document}


\renewcommand{\baselinestretch}{2}

\markright{ \hbox{\footnotesize\rm 
}\hfill\\[-13pt]
\hbox{\footnotesize\rm
}\hfill }

\markboth{\hfill{\footnotesize\rm FIRSTNAME1 LASTNAME1 AND FIRSTNAME2 LASTNAME2} \hfill}
{\hfill {\footnotesize\rm A Robust Consistent Information Criterion For Model Selection} \hfill}

\renewcommand{\thefootnote}{}
$\ $\par


\fontsize{12}{14pt plus.8pt minus .6pt}\selectfont \vspace{0.8pc}
\centerline{\large\bf A Robust Consistent Information Criterion for Model Selection }
\vspace{2pt} \centerline{\large\bf based on Empirical Likelihood}
\vspace{.4cm} \centerline{Chixiang Chen$^1$, Ming Wang$^{1}$, Rongling Wu$^1$, Runze Li$^2$} \vspace{.4cm} \centerline{\it
$^1$ Division of Biostatistics and Bioinformatics, Department of Public Health Science,}\centerline{\it Pennsylvania State College of Medicine, Hershey, PA, U.S.A., 17033}\centerline{\it
$^2$ Department of Statistics and the Methodology Center,}\centerline{\it Pennsylvania State University, University Park, PA, U.S.A, 16802} \vspace{.55cm} \fontsize{9}{11.5pt plus.8pt minus
.6pt}\selectfont


\begin{quotation}
\noindent {\it Abstract:}\\
Conventional likelihood-based information criteria for model selection rely on the distribution assumption of data. However, for complex data that are increasingly available in many scientific fields, the specification of their underlying distribution turns out to be challenging, and the \red{existing criteria may be limited and are not general enough to handle a variety of model selection problems. 
Here, we propose \red{a robust and consistent model selection criterion} based upon the empirical \blue{likelihood function which is data-driven. In particular, this framework adopts plug-in estimators that can be achieved by solving external estimating equations, not limited to the empirical likelihood, which avoids potential computational convergence issues and allows versatile applications, such as generalized linear models, generalized estimating equations, penalized regressions and so on. The formulation of our proposed criterion is initially derived from the asymptotic expansion of the marginal likelihood under variable selection framework, but more importantly, the consistent model selection property is established under a general context.} Extensive simulation studies confirm the out-performance of the proposal compared to traditional model selection criteria. Finally, an application to the Atherosclerosis Risk in Communities Study illustrates the practical value of this proposed framework.}

\vspace{9pt}
\par
\end{quotation}\par

\def\thefigure{\arabic{figure}}
\def\thetable{\arabic{table}}

\renewcommand{\theequation}{\thesection.\arabic{equation}}

\fontsize{12}{14pt plus.8pt minus .6pt}\selectfont

\setcounter{section}{1} 
\setcounter{equation}{0} 
\noindent {\bf 1. Introduction}

Model selection is a common problem encountered in various disciplines including variable selection in the mean structure, correlation structure selection for longitudinal data analysis, and tuning parameter selection in penalized regression among others. Currently, commonly used approaches for model selection rely on several likelihood-based information criteria, such as Akaike information criterion (AIC) \citep{akaike1974}, Bayesian information criterion (BIC) \citep{schwarz1978}, and Generalised Information Criteria (GIC) \citep{konishi1996}. \red{However, such information criteria critically depend upon the parametric distribution assumption, and have limited applications in more complicated model selection problems rather than variable selection \citep{chen2012}. More importantly, the distribution misspecification would have a negative impact on the selection performance and be inevitably encountered in practice when the data are complex and heterogeneous. For instance, in some survey studies, some variables such as Beck's depression index or caffeine/alcohol use could be highly skewed or over-dispersed due to sampling bias, thus a well-defined distribution is usually difficult to identify. However, these complex data, increasingly available due to advances in data collection techniques, play so critical roles in capturing fundamental principles underlying natural, social, and engineering processes that more advanced and rigorous approaches need to be adopted for valid inference.}

To avoid the distribution specification but still borrow the likelihood properties, a data-driven approach based on the empirical likelihood has been developed \citep{owen1988,qin1994}, and has been widely applied for data analysis and statistical inference \citep{owen2001}. However, the empirical likelihood-based information criteria for model selection are still not well studied in past decades. \cite{eic1995} first proposed empirical information criterion (EIC) based on the Kullback-Leibler distance between discrete empirical distributions. However, there is a severe convergence issue for empirical likelihood estimators to compute this criterion, thus limiting the application. \blue{To avoid computation issue, \cite{variyath2010} advocated an empirical AIC and an empirical BIC based on the adjusted empirical likelihood by incorporating an extra parameter \citep{chen2008adjusted, chen2020missing}. \red{However, such an approach only focuses on variable selection in the mean structure and still requires empirical likelihood estimators, and also how to select the nuisance parameter remains an issue for practical use.} We also refer to the literature where empirical likelihood-based criteria have been proposed in particular situations but without any theoretical justification \citep{chen2012, tang2010, chang2018}. To the best of our knowledge, few studies are investigated for empirical likelihood-based information criteria that can be broadly applicable to the general model selection not limited to variable selection.}

\red{Here, we consider a general model selection context with a collection of candidate models $M_1$, $M_2,\ldots, M_k$ with the true model included. The main objective is how to capture the true one among the candidates. Under the Bayesian paradigm with non-informative prior, the marginal likelihood of model selection is of main focus and shown as 
\begin{equation}\label{posterior}
P(\bfD|M)=\int P(\bfD|\bfgamma, M) P(\bfgamma|M)d\bfgamma \propto \int P(\bfD|\bfgamma) d\bfgamma,
\end{equation}    
\chixiang{where $\bfD$ is denoted as the full data.} $\bfgamma$ are the parameters in the candidate model $M$, and $P(\bfD|\bfgamma)$ is indeed the likelihood function $L(\bfgamma|\bfD)$. When $L(\bfgamma|\bfD)$ is fully specified, the well-known criterion $\textsc{BIC}=-2\log L(\hat{\bfgamma}|\bfD)+p\log n$ has been derived by selecting the model corresponding to the largest marginal probability in (\ref{posterior}). However, several restrictions exist: 1) the distribution needs to be pre-specified for the likelihood $L$ and 2) the estimator $\hat{\bfgamma}$ is MLE. In this work, we will present a robust and consistent information criterion, named empirical likelihood-based consistent information criterion (ELCIC), which targets model selection under the general contexts not limited to variable selection. In particular, the robustness indicates the distribution-free property and flexibility in wide application to the general model selection problems; the consistency is defined as the capture of the true model with probability tending to one, an independent and hot ongoing research topic in past decades \citep{ebic2008,variyath2010,kim2016consistent}. The likelihood part in ELCIC is purely data-driven based on the empirical likelihood, and we relax the procedure for parameter estimation by utilizing plug-in estimators to calculate this criterion, under which the consistency property can still hold under mild conditions. 

The rest of the paper is organized as follows. In Section \ref{section2}, to demonstrate the formulation of our proposed criterion, we provide the theoretical derivation under the variable selection framework by asymptotically expanding the marginal probability $(\ref{posterior})$. However, the consistency selection property of our proposed criterion is investigated under more general model selection settings under mild conditions, which is of the most importance and main goal in this article. Section \ref{simulation} considers three specific cases for illustration and evaluates the finite-sample performances via simulation studies. We further apply our proposal by a real data example in Section \ref{section4}. Last but not least, several promising extensions as future works are discussed in Section \ref{discussion}.}

\section{Methodology}\label{section2}
\subsection{Empirical Likelihood}    
Based upon the idea from the empirical distribution, \cite{owen1988} introduced an empirical likelihood approach to construct the likelihood-based confidence intervals. The full data is denoted by $\bfD=\{\bfD_i\}_{i=1}^{n}$ with $\bfD_i=(\bfX_i^\T,\bfY_i^\T)^\T$, assumed to be independent and identically distributed (i.i.d.) when a regular regression is considered. Given some estimating equations $\bfg(\bfD_i,\bfgamma)$ satisfying $E\bfg(\bfD_i,\bfgamma)=\mathbf{0}$ for $i=1,\ldots, n$, the empirical likelihood ratio is defined by
\begin{equation}\label{el}
R^F=\sup_{\bfgamma,p_1,\ldots, p_n}\left\{ \prod_{i=1}^{n} np_i; p_i\geq0, \sum_{i=1}^{n}p_i=1, \sum_{i=1}^{n}p_i\bfg\big(\bfD_i,\bfgamma\big)=\mathbf{0} \right\}.
\end{equation}
\red{Compared to the traditional likelihood, point mass probabilities for the observations are utilized instead here, thus the information from the data would be automatically and efficiently borrowed from the estimating equations constrain \citep{qin1994}, which is the desired property and shows its great potential for model selection. Given the estimator denoted by $\hat{\bfgamma}$ which will be discussed more next,
the negative logarithm of empirical likelihood ratio can be easily calculated below based on the Lagrange multiplier method \citep{owen2001}, \blue{
\begin{equation}\label{logl}
l=-\log R^F(\hat{\bflambda},\hat{\bfgamma})=\sum_{i=1}^{n}\log\{1+\hat{\bflambda}^\T\bfg(\bfD_i, \hat{\bfgamma})\},
\end{equation}}
where the parameter estimate $\hat{\bflambda}$ can be obtained by solving the following equations based on the Newton-Raphson method
\begin{equation}\label{lambda}
\frac{1}{n}\sum_{i=1}^{n}\frac{\bfg\big(\bfD_i,\hat{\bfgamma}\big)}{1+\boldsymbol{\lambda}^\T\bfg\big(\bfD_i,\hat{\bfgamma}\big)}=\mathbf{0}.
\end{equation}}
\vspace{-0.5in}
\subsection{Derivation of ELCIC under the Variable Selection Framework}\label{vs}
\blue{Before introducing our proposed model selection criterion and discussing its consistency, let us start with the variable selection in a regression framework to get some insights on its formulation.} Here, $\bfgamma$ is denoted as the parameter vector in the mean structure only. To implement the selection relying on the empirical likelihood, we need to specify a full set of estimating equations \citep{eic1995,variyath2010,chen2012}. For any $i=1,2,\ldots,n$, we define
\begin{equation}\label{g1g2}
\bfg(\bfD_i,\bfgamma)=\left(
\begin{matrix}
\bfg_1(\bfD_i,\tilde{\bfgamma})   \\
\bfg_2(\bfD_i,\tilde{\bfgamma})  \\
\end{matrix}
\right),
\end{equation}
where $\tilde{\bfgamma}=(\bfgamma^\T, \mathbf{0}^\T)^\T$ so that its dimension will be matched with the pre-specified full covariate matrix $\bfX_i$, and $\bfg_1(\bfD_i,\tilde{\bfgamma})$ and $\bfg_2(\bfD_i,\tilde{\bfgamma})$ correspond to the estimating equations for the parameters with and without involvement in a candidate model, respectively. In variable selection, we denote the cardinality of $\tilde{\bfgamma}$ as $L$ and that of $\bfgamma$ from a candidate model as $p$ with $0<p\leq L<\infty$. To be noted that the formula in (\ref{g1g2}) is constructed only for implementing the variable selection.

As aforementioned, we will encounter the computation issue when maximizing (\ref{el}) to obtain the empirical likelihood-based estimator denoted by $\hat{\bfgamma}_{EL}$. \red{It requires zero be inside of the convex hull of estimating equations to ensure the existence of solutions \citep{qin1994,chen2008adjusted}. To overcome this computation issue and make the criterion more versatile accommodating to various needs, we consider plug-in estimators instead, which in the variable selection framework are obtained from solving some external estimating equations $\sum^n_{i=1}\bfg_1(\bfD_i,\bfgamma)=\mathbf{0}$ in (\ref{g1g2}) to bypass the complex and also unstable estimation procedure from (\ref{el}). Given the plug-in estimators denoted as $\hat{\bfgamma}_{EE}$ and the corresponding Lagrange multiplier estimator \blue{$\hat{\bflambda}_{EE}$}, the negative logarithm of empirical likelihood ratio in (\ref{logl}) can be achieved.} Instead of utilizing the likelihood under a pre-specified distribution to maximize the marginal likelihood (\ref{posterior}), we employ the empirical likelihood ratio for $L(\bfgamma|\bfD)$ shown as below
\begin{equation}
L(\bfgamma|\bfD)=R^F(\bflambda,\bfgamma)=\prod_{i=1}^{n}\{1+\bflambda^\T \bfg(\bfD_i,\bfgamma)\}^{-1}.
\end{equation}

\redreview{Next, we provide two conditions to facilitate the asymptotic expansion of $P(\bfD|M)$. Note that we use $\lVert\cdot\rVert$ to denote the Euclidean norm and $\lvert\cdot\rvert$ as a determinant of a matrix. Also, for notation simplicity, the subscript for $\bfD$ is disregarded below due to the i.i.d. property, and also the expectations are all evaluated at the true parameter values under the correctly specified models.}

\begin{condition}\label{c1} (Regularities)
	Given the correctly specified model with the true parameter $\bfgamma_0$ \redreview{satisfying $E\big\{\bfg(\bfD, \bfgamma_0)\big\}=\mathbf{0}$}, \redreview{$E\big\{\bfg(\bfD, \bfgamma_0)\bfg^\T(\bfD, \bfgamma_0)\big\}$} is positive definite, and $\{\partial^2 \bfg(\bfD, \bfgamma)\}/(\partial\bfgamma^\T\partial \bfgamma)$ is continuous in the neighborhood of $\bfgamma_0$. Furthermore,  we assume $\lVert \{\partial\bfg(\bfD, \bfgamma)\}/(\partial\bfgamma^\T)\rVert$, $\lVert \{\partial^2 \bfg(\bfD, \bfgamma)\}/(\partial\bfgamma^\T \partial\bfgamma)\rVert$, and $\lVert  \bfg(\bfD, \bfgamma)\rVert^3$ are bounded by some integrable function around $\bfgamma_0$. 
\end{condition}

\begin{condition}\label{c2} (Efficiency)
	For the estimating equations $\bfg_1$ and $\bfg_2$ defined in (\ref{g1g2}) and given \redreview{the correctly specified model with the true parameter $\tilde{\bfgamma}_0=(\bfgamma_0^\T, \mathbf{0}^\T)^\T$}, \redreview{
	\begin{equation*}
	\begin{split}
	E\Big\{\frac{\partial\bfg_1(\bfD, \tilde{\bfgamma}_0)}{\partial\bfgamma^\T}\Big\}&=-E\big\{\bfg_1(\bfD, \tilde{\bfgamma}_0)\bfg_1^\T(\bfD, \tilde{\bfgamma}_0)\big\}, \\ 
	E\Big\{\frac{\partial\bfg_2(\bfD, \tilde{\bfgamma}_0)}{\partial\bfgamma^\T}\Big\}&=-E\big\{\bfg_1(\bfD, \tilde{\bfgamma}_0)\bfg_2^\T(\bfD, \tilde{\bfgamma}_0)\big\}.
	\end{split}
	\end{equation*}}
\end{condition}
\textit{Condition} \ref{c1} includes several regular moment conditions to ensure valid inference based on empirical likelihood \citep{qin1994}. To simplify the formula of our proposed criterion, \textit{Condition} \ref{c2} sets some constrains on these two estimating equations $\bfg_1$ and $\bfg_2$ in (\ref{g1g2}), which are related to the estimator efficiency. Note that \textit{Condition} \ref{c2} contains a family of estimating equations leading to asymptotically efficient estimators. For instance, when the score function is utilized, \textit{Condition} \ref{c2} is definitely satisfied by the property of \redreview{the Fisher information under regular conditions} \citep{pierce1982}; when $\bfg$ is generalized estimating equations with a correctly specified correlation structure \citep{liang1986}, \textit{Condition} \ref{c2} holds as well. However, \textit{Condition} \ref{c2} is not required for the proof of model selection consistency, which will be discussed more later.

\begin{theorem}\label{thm1}
	\redreview{Under \textit{Conditions} \ref{c1} and \ref{c2}, given $\hat{\bfgamma}_{EE}$ obtained from the estimating equations $\bfg_1$ in (\ref{g1g2}) and the rank of $E[\{\partial\bfg(\bfD, \bfgamma_0)\}/(\partial\bfgamma^\T)]$ as $p$ the same as the dimensionality of $\hat{\gamma}_{EE}$, and by applying the Laplace approximation and setting \redreview{a} non-informative prior to $\bfgamma$, we have
	\begin{equation*}
	\chixiang{-2\log P(\bfD|M)= -2\log R^F(\hat{\bflambda}_{EE},\hat{\bfgamma}_{EE})+p\log n+\tilde{C}+o_p(1),}
	\end{equation*}
   where
	$\tilde{C}=\log(\bfSigma_{21}\bfSigma^{-1}_{11}\bfSigma_{12})-p\log(2\pi)-2\log(\tilde{A})$
	with $\bfSigma_{11} =E\big\{\bfg(\bfD, \bfgamma_0)\bfg^\T(\bfD, \bfgamma_0)\big\}$;
	$\bfSigma_{12} =E\big[\partial\bfg(\bfD, \bfgamma_0)/\partial\bfgamma^\T\big]$,  $\bfSigma_{21} =\bfSigma_{12}^\T$;  $\tilde{A}=\int\exp\{(1/2)\bfdelta_1^\T(n\bfSigma_{11})\bfdelta_1\}\rho_{\bfdelta_1}(\bfdelta_1)d\bfdelta_1$; and also, $\rho_{\bfdelta_1}(\cdot)$ is some prior function of the random variable $\bfdelta_1$ defined in the Supplementary Material.} 
\end{theorem}
Based upon Theorem $\ref{thm1}$, our proposed ELCIC is defined by 
\begin{equation}
\textsc{ELCIC}=-2\log R^F(\hat{\bflambda}_{EE},\hat{\bfgamma}_{EE})+p\log n. 
\end{equation}
\chixiang{Noting that the proposed ELCIC is free of prior specification, a desired property for a well-defined model selection criterion, with finite values guaranteed regardless of the fact that $\tilde{A}$ in Theorem \ref{thm1}
might be infinity for some prior functions.}

 \redreview{Theorem $\ref{thm1}$ is derived under the framework of variable selection based on the Laplace approximation with some extra constraints on estimating equations specified in \textit{Condition} \ref{c2}. However, in next section, we will show that \textsc{ELCIC} is consistent, i.e., it can capture the true model with the probability approaching to one, and those conditions can be relaxed. Note that this consistency property holds under any general estimating equations, not limited to the ones defined in (\ref{g1g2}). We will rigorously explore the consistency of our \textsc{ELCIC} next. }

\subsection{Model selection consistency of \textsc{ELCIC}}
\redreview{In this section, we will focus on our proposal's consistency of general model selection not limited to variable selection.} \blue{Under this general context, let us redefine \redreview{$\bfg(\bfD,\bfgamma)$} as some full estimating equations in (\ref{el}) satisfying \textit{Condition} \ref{c1}, where the $p$-by-1 parameter vector $\bfgamma$ includes all the parameters in the estimating equations \redreview{$\bfg(\bfD,\bfgamma)$}, such as parameters in the mean structure, coefficients in the correlation matrix, or any other nuisance parameters. Notice that we do not assume that \redreview{$\bfg(\bfD,\bfgamma)$} correspond to the estimating equations for all the parameters in $\bfgamma$}. We further relax the assumption that the plug-in estimators $\hat{\bfgamma}_{EE}$ could be derived from other external estimating equations, not necessarily specific to (\ref{g1g2}), but with some mild conditions satisfied (i.e., \textit{Conditions} \ref{c5}). More concrete examples for specifying the function \redreview{$\bfg(\bfD,\bfgamma)$} will be discussed in Section \ref{case}.
\begin{theorem}\label{thm2}
	Under \textit{Condition} \ref{c1} and given the candidate model $M$ is correctly specified with the true value of $\bfgamma$ as $\bfgamma_0$,
		 \chixiang{let us} define $\bfQ_n= (1/n) \sum_{i=1}^n \bfg(\bfD_i,  {\bfgamma}_{0})$ $+ \bfSigma_{12}(\hat{\bfgamma}_{EE} -\bfgamma_{0})$ with plug-in estimators satisfying $\hat{\bfgamma}_{EE} -\bfgamma_{0}=O_p(\bfn^{-1/2})$, then the negative logarithm of the likelihood is
		\begin{equation}\label{loglike}
		l=\frac{1}{2} (\redreview{n^{1/2}}\bfQ_n^\T)\bfSigma^{-1}_{11} (\redreview{n^{1/2}}\bfQ_n)+o_p(1)
		\end{equation}
\end{theorem}

Theorem \ref{thm2} provides an insight of the order of $l$ when the model is correctly specified, which will be a crucial component for the derivation of the consistency of our proposed criterion. \blue{As a by-product of \chixiang{Theorem \ref{thm2}}, the Corollary below provides the asymptotic distribution of $2l$ in the variable selection framework with $\hat{\bfgamma}_{EE}$ from Section \ref{vs}.} 
\begin{corollary}\label{corollary}
 \redreview{Given the same conditions in Theorem \ref{thm2}, $\bfg$ and $\bfg_1$ defined in (\ref{g1g2}) in the case of variable selection and $\tilde{\bfgamma}_0= (\bfgamma_0^\T,\mathbf{0}^\T)^\T$, $\hat{\bfgamma}_{EE}$ satisfies}
		\begin{equation}\label{plug in}
		\hat{\bfgamma}_{EE}-\bfgamma_{0}=-\bigg\{\redreview{E\Big(\frac{\partial\bfg_1(\bfD,\tilde{\bfgamma}_0)}{\partial\bfgamma^\T}\Big)}\bigg\}^{-1}\frac{1}{n} \sum_{i=1}^n \bfg_1(\bfD_i,\tilde{\bfgamma}_{0})+o_p(\bfn^{-\frac{1}{2}}).  
		\end{equation}
		Then, we obtain $2l$ converges in distribution to $\sum_{j=1}^{\tilde{L}}\Lambda_j\chi_1^2$, where $\Lambda_1, \ldots,\Lambda_ {\tilde{L}}$ are non-zero eigenvalues of matrix $\bfOmega= \bfSigma_{11}^{1/2} \bfSigma_{\ast}^\T \bfSigma^{-1}_{11} \bfSigma_{\ast}\bfSigma_{11}^{1/2}$ with $\tilde{L}=\text{rank} (\bfOmega)$ and
		$\bfSigma_{\ast}= \bfI_{L\times L}-\Big(\bfSigma_{12}\redreview{\{E(\partial\bfg_1(\bfD,\tilde{\bfgamma}_0)/\partial \bfgamma^\T\big)\}^{-1}}, \boldsymbol{0}_{L\times(L-p)} \Big)$.
\end{corollary}

\begin{condition}\label{c3} (Regularity for Misspecified Model)
	Let $\bfgamma_*\neq\bfgamma_0$, then for any $\bfgamma$ in the neighbourhood of $\bfgamma_*$, we have $E\lVert \bfg(\bfD, \bfgamma)\rVert^{2+\delta} < \infty$ with some $\delta>0$.
\end{condition}
\begin{condition}\label{c4} (Identifiability)
	For any $\bfgamma$ in the neighbourhood of $\bfgamma_*\neq\bfgamma_0$, we have the following condition: $\lVert E \bfg(\bfD, \bfgamma)\rVert>0$.
\end{condition}

\textit{Condition} \ref{c3} is an extension of regularities in \textit{Condition} \ref{c1} when the candidate model $M$ is mis-specified. We require that the $(2+\delta)^{th}$ moment of estimating equations is finite. \textit{Condition} \ref{c4} is the identifiability assumption, which further implies that the model is identifiable if only the correctly specified model has $E \bfg(\bfD, \bfgamma_0)=\mathbf{0}$ satisfied. This is also the key to model selection.

\begin{theorem}\label{thm3}
	Under \textit{Conditions} \ref{c1}, \ref{c3} and \ref{c4}, for any $\bfgamma$ in the neighbourhood of $\bfgamma_*\neq\bfgamma_0$, we have $n^{1-c} \lVert\bar{\bfg}_n\rVert^2\log(n)l^{-1}=O_p(1)$, where $\bar{\bfg}_n=(1/n) \sum_{i=1}^n \bfg(\bfD_i,  \bfgamma)$ for $\frac{1}{2}<c<1$.
\end{theorem}

Theorem \ref{thm3} ensures that when the candidate model is misspecified, the negative log likelihood $l$ will tend to be infinity with the order at least $\log n$. Together with Theorem \ref{thm2} and Theorem \ref{thm3}, and the following \textit{Condition} \ref{c5}, we are ready to present the main result of this paper.

\begin{condition}\label{c5} (Well-behaved Estimator)
	Denote 
	$\bfV= \textbf{Cov} \big((1/n) \sum_{i=1}^n \bfg(\bfD_i, \bfgamma_0)$ $+\bfSigma_{12}(\hat{\bfgamma}_{EE}-\bfgamma_0)\big)$ with \redreview{$\hat{\bfgamma}_{EE} -\bfgamma_{0}=O_p(\bfn^{-1/2})$}. We have $tr(\bfSigma_{11}^{-1}\bfV)<\infty$, given the correctly specified candidate model.
\end{condition}

\textit{Condition} \ref{c5} is a mild condition requiring well-behaved plug-in estimators. If this plug-in estimator $\hat{\bfgamma}_{EE}$ is from Corollary \ref{corollary}, then $tr(\bfSigma_{11}^{-1}\bfV)<\infty$ in \textit{Condition} \ref{c5} will be equivalent to ask for finite eigenvalues of $\bfOmega$ defined in Theorem \ref{thm2}. On the other hand, when $\bfg(\bfD,\bfgamma)$ are the estimating equations for $\bfgamma$, and the estimator $\hat{\bfgamma}_{EL}$ are obtained from maximizing empirical likelihood ratio (\ref{el}), we will have $tr(\bfSigma_{11}^{-1}\bfV)=L-p$, which satisfies \textit{Condition} \ref{c5}.

\begin{theorem}\label{thm4}
	Under \textit{Conditions} \ref{c1}, \ref{c3}, \ref{c4} and \ref{c5} and given the true model denoted by $M_0$,  
	$P\big[\min\{\textsc{ELCIC}(M):M \neq M_0\}>\textsc{ELCIC}(M_0)\big]\rightarrow 1
	\text{ as } n\rightarrow \infty$. 
\end{theorem}

\red{The proof strategy is standard large sample theory; however, there are many important findings from Theorem \ref{thm4} revealing the underlying merits and implications of ELCIC for the general model selection purpose. \redreview{First, the theorem holds under very mild conditions, in particular, not involving \textit{Condition} \ref{c2} and the Laplace approximation. Moreover, the proof does not rely on a very specific form of full estimating equations, not necessarily limited to the estimating equations $\bfg(\bfD_i,\bfgamma)$ in (\ref{g1g2}). It implies the potential of ELCIC to various model selection problems, rather than solo variable selection. Second, there are few restrictions for the plug-in estimators except \textit{Condition} \ref{c5}. In practice, we might apply some common estimation procedures such as least square, score functions, generalized estimating equations, loss functions and so on, making ELCIC more flexible and versatile. It is noted that the parameters $\bfgamma$ are not only limited to primary parameters of interest in the candidate model but also could include other nuisance parameters, such as the correlation coefficients in the generalized estimating equation method for longitudinal data and the parameters in logistic regressions for observing probabilities in the inverse probability weights method. Therefore, the consistency property in Theorem \ref{thm4} allows ELCIC to deal with broad model selection problems with three case studies shown next.}}


\section{Case Studies and Numeric Results}\label{simulation}
\subsection{Full Estimating Equations} \label{case}
We will discuss how to specify full estimating equations $\bfg(\bfD_i,\bfgamma)$ for ELCIC in three cases. To avoid confusion, we denote the parameter vector in the mean structure as $\bfbeta$ and the overall parameter vector in the estimating equations as $\bfgamma$. 

\textbf{Case 1: Generalized Linear Models (GLM).} \cite{glm} introduced the GLM concept to unify the theory for different models in categorical analysis. In this case, the full estimating equations $\bfg$ in (\ref{el}) can be simply defined as the score functions:
\begin{equation}\label{poi}
\bfg(\bfD_i,\bfbeta)=\bfX_i(Y_i-\mu_i(\tilde{\bfbeta}))
\end{equation}
where $\mu_i(\tilde{\bfbeta})$ with $\tilde{\bfbeta}=(\bfbeta^\T,\mathbf{0}^\T)^\T$ is the conditional expectation of $Y_i$ modeled by $f(\bfX_i^T\tilde{\bfbeta})$ with some pre-specified canonical link function $f$. Since the equation (\ref{poi}) is valid only when the mean structure is correctly specified and does not require the second moment, ELCIC under the full estimating equations (\ref{poi}) is capable to handle the scenario when the variance structure is mis-specified, such as over-dispersion, which is often encountered in count data analysis. Only variable selection is of research interest in this case.

\textbf{Case 2: Generalized Estimating Equations (GEE).} Now we extend our focus to model selection for longitudinal data. \cite{liang1986} introduced the marginal model to conduct statistical inference without specifying the joint distribution of longitudinal data. To be noted, a correctly specified mean structure is always a key for estimation consistency, and in the meantime, it can further gain the efficiency by identifying correct ``working" correlation structure in GEE approach \citep{mccullagh1989}. In this case, we will properly specify the full estimating equations so that ELCIC can simultaneously select marginal mean and correlation structures, while the main existing criteria, such as quasi-likelihood criteria (QIC) \citep{pan2001akaike}, cannot handle the joint selection.   

To achieve our goal and for simplicity, let us assume a balanced design with $T$ observations for each subject. For subject $i$, the marginal mean is denoted as $\bfmu_i$ and variance-covariance matrix as $\bfV_i$. The over-dispersion parameter is noted by $\phi$ (assumed known but can also be consistently estimated), and the correlation coefficient vector as $\bfrho^c=(\rho_1^c,\ldots,\rho_{T-1}^c)^{\T}$. Here, the superscript $c$ indicates the type of correlation structure. For instance, under a stationary structure, we have $\bfrho^{STA}=(\rho_1^{STA},\ldots,\rho_{T-1}^{STA})^{\T}$. Thus, the full estimating equations in (\ref{el}) can be defined as 
\begin{align}\label{longfull}
\bfg\Big(\bfD_i,\bfbeta,\bfrho^c\Big)=&
\begin{pmatrix} 
\bfH_i^\T \bfV_i^{-1}\Big(\bfY_i-\bfmu_i(\tilde{\bfbeta})\Big)\\  
\bfU_i(\tilde{\bfbeta})-\boldsymbol{h}(\bfrho^c)\phi
\end{pmatrix},
\end{align}
where $\tilde{\bfbeta}$ is defined as $(\bfbeta^\T,\mathbf{0}^\T)^\T$, $\bfH_i$ denotes the first derivative of $\bfmu_i$ with respect to $\tilde{\bfbeta}$; and $\bfU_i(\tilde{\bfbeta})=\big(U_{i1}(\tilde{\bfbeta}), U_{i2}(\tilde{\bfbeta}),\ldots, U_{i(T-1)}(\tilde{\bfbeta})\big)^\T$ with
\begin{equation}\label{stationary}
U_{im}(\tilde{\bfbeta})=\sum_{j=1}^{T-m}e_{ij}(\tilde{\bfbeta})e_{i,j+m}(\tilde{\bfbeta}),~ \text{for }~m=1,\ldots,T-1.
\end{equation}
Also, $e_{ij}$ represents the standardized residual term $(y_{ij}-\mu_{ij})/\surd{\nu_{ij}}$ , for $i=1,\ldots,n$ and $j=1,\ldots, T$. Finally, $\boldsymbol{h}(\bfrho^c)$ is defined as $\big(\rho_1^c\big(T-1-p/n\big),\ldots,$ $ \rho_{T-1}^c(1-p/n)\big)^\T$.

To be noted, $\tilde{\bfbeta}$ is proposed to achieve variable selection in marginal mean structure, and a stationary correlation structure is utilized based on (\ref{stationary}) to implement the correlation structure selection. Thus, ELCIC can simultaneously select marginal mean and correlation structures since only when both structures are correctly specified, the expectation of full estimating equations (\ref{longfull}) equal zeros. Of note is that \cite{chen2012} proposed empirical likelihood BIC to select ``working" correlation structure alone, under which the full estimating equations is the subset of our proposed one (\ref{longfull}). Accordingly, ELCIC unifies the selection procedure by allowing for both marginal mean and correlation structures. Besides, the theorems in the Supplementary Material further provide a theoretical justification of this criterion, which is lacking in current literature. 

\textbf{Case 3: Penalized Generalized Estimating Equations (PGEE).} Penalized regression is one of the most popular research topics in the past two decades \citep{lasso1996, scad}. It utilizes some penalties to shrink the effect of unnecessary features toward zero by identifying a proper tuning parameter. In this case, we mainly focus on the tuning parameter selection. There are two common ways to choose the tuning parameter: cross-validation or some BIC-type methods \citep{ebic2008}. It is well known that cross-validation will lead to high false positive, while the BIC-type method is less time consuming and tends to have a lower false positive rate. However, BIC-type criteria cannot be applied to semi-parametric or non-parametric contexts. Here, we consider PGEE proposed by \cite{wanglan2012}, where BIC is no longer suitable, but ELCIC can be easily embedded into this framework.  

\textsc{PGEE} is the combination of \textsc{GEE} and the first derivative of smoothly clipped absolute deviation (\textsc{SCAD}) penalty, which could facilitate the sparsity in the marginal mean structure. The selection consistency and asymptotic normality were investigated in \cite{wanglan2012} with a diverging number of covariates. Here, we only consider the cases with a fixed number of covariates. The full estimating equations $\bfg(\bfD_i,\bfgamma)$ in (\ref{longfull}) will be utilized to implement model selection. To be noted again, the consistency of \textsc{ELCIC} doesn't require the estimating equations for $\hat{\bfbeta}_{EE}$ to be contained in $\bfg(\bfD_i,\bfgamma)$, which theoretically justifies the application of \textsc{ELCIC} to penalized regression cases. Moreover, based on the selection consistency and asymptotic normality in \cite{wanglan2012}, there exists a tuning parameter which will correctly identify true zeros with probability tending to one and also make the non-zero part of the estimators converge to the true one with order $O_p(n^{-\frac{1}{2}})$ under fixed $p$. Therefore, we can apply \textsc{ELCIC} to locate such ``optimal" tuning parameter. What's more, the same rationale to Case 2, the full estimating equations $\bfg(\bfD_i,\bfgamma)$ in (\ref{longfull}) could facilitate joint selection of marginal mean and ``working" correlation structures, which is unfeasible based on cross-validation (CV).

\redreview{Besides the three traditional cases we presented above, ELCIC has the potential to deal with the complicated scenarios where the existing criteria might not fit well or even not be applicable at all. For instance, the variable selection for the augmented inverse probability weighting (AIPW) method that is commonly used for missing data analysis \citep{robins1994estimation} with extensive work in longitudinal data, survival analysis and causal inference \citep{bang2005doubly,seaman2009doubly,scharfstein1999adjusting,long2011}. We provide the detailed discussion for such cases with numerical evaluation of our proposal, which are displayed in the Supplementary Material. }  

\subsection{Numeric Results}\label{section 4}
In this subsection, we will conduct simulation studies under the three cases in Section \ref{case}, and compare ELCIC with some popular existing criteria in each case to show the robustness of our proposed criteria. \redreview{Due to space limit, the extra simulation studies for variable selection under the AIPW framework is presented in the Supplementary Material.}

\begin{table}[ht]
	\centering
	\caption{Performance of ELCIC compared with AIC and BIC for the scenarios under Poisson distribution with potential over-dispersed outcomes. 500 Monte Carlo data are generated with sample size $n=100, 200$. The model with $\{x_1,x_2\}$ is the true one. \redreview{NB: Negative bionomial with $k$ represented as the number of failures} }
	\scalebox{0.85}{%
	\begin{minipage}[c]{1.3\textwidth}
	\begin{tabular}{lllccccccc}
		\toprule
		\multirow{2}{*}{Sample Size} & \multirow{2}{*}{Distribution} & \multirow{2}{*}{Criteria} & \multicolumn{7}{c}{Candidate Models} \\ 
		&                     &                         & $x_1$  & $x_2$  & $x_3$ & \boldsymbol{$x_1,x_2$}      & $x_1,x_3$ & $x_2,x_3$ & $x_1,x_2,x_3$ \\ \midrule
		$n=100$ & POISSON             & AIC                     & 0     & 0     & 0    & 0.852          & 0         & 0         & 0.148          \\
		&                     & GIC                     & 0     & 0     & 0    & 0.798  & 0         & 0         & 0.202           \\
		&                     & BIC                     & 0     & 0     & 0    & {0.980}  & 0         & 0         & 0.020           \\
		&                     & ELCIC                    & 0     & 0     & 0    & 0.95           & 0         & 0         & 0.050           \\
		& NB $k=8$       & AIC                     & 0     & 0     & 0    & 0.786          & 0.002     & 0         & 0.212          \\
		&                     & GIC                     & 0     & 0     & 0    & 0.792  & 0         & 0         & 0.206          \\
		&                     & BIC                     & 0     & 0     & 0    & 0.926          & 0.002     & 0         & 0.072          \\
		&                     & ELCIC                    & 0.002 & 0.002 & 0    & {0.940}  & 0.002     & 0         & 0.054          \\
		& NB $k=2$       & AIC                     & 0.002 & 0.002 & 0    & 0.592          & 0.006     & 0.002     & 0.396          \\
		&                     & GIC                     & 0.004     & 0.008    & 0    & 0.714  & 0.012         & 0.006         & 0.256           \\
		&                     & BIC                     & 0.004 & 0.008 & 0    & 0.774          & 0.012     & 0.002     & 0.200            \\
		&                     & ELCIC                    & 0.022 & 0.052 & 0    & {0.850}  & 0.018     & 0.002     & 0.056          \\ \midrule
		$n=200$ & POISSON & AIC  & 0     & 0     & 0    & 0.830           & 0         & 0         & 0.170           \\
		&                     & GIC                     & 0     & 0     & 0    & 0.800  & 0         & 0         & 0.200           \\
		&                     & BIC                     & 0     & 0     & 0    & {0.978} & 0         & 0         & 0.002          \\
		&                     & ELCIC                    & 0     & 0     & 0    & 0.962          & 0         & 0         & 0.038          \\
		& NB $k=8$      & AIC                     & 0     & 0     & 0    & 0.718          & 0         & 0         & 0.282          \\
		&                     & GIC                     & 0     & 0     & 0    & 0.776  & 0         & 0         & 0.224           \\
		&                     & BIC                     & 0     & 0     & 0    & 0.916          & 0         & 0         & 0.084          \\
		&                     & ELCIC                    & 0     & 0     & 0    & {0.952} & 0         & 0         & 0.048          \\
		& NB $k=2$       & AIC                     & 0     & 0     & 0    & 0.562          & 0.002     & 0         & 0.436          \\
		&                     & GIC                     & 0     & 0     & 0    & 0.764  & 0.002         & 0         & 0.234           \\
		&                     & BIC                     & 0     & 0     & 0    & 0.814          & 0.002     & 0         & 0.184          \\
		&                     & ELCIC                    & 0.002 & 0     & 0    & {0.946} & 0.004     & 0         & 0.048\\ \midrule
		$n=400$	&	POISSON	&	AIC &	0	&	0	&	0	&	0.848	&	0	&	0	&	0.152	\\
		&		&	GIC 	&	0	&	0	&	0	&	0.810	&	0	&	0	&	0.190	\\
		&		&	BIC &	0	&	0	&	0	&	0.990	&	0	&	0	&	0.010	\\
		&		&	ELCIC 	&	0	&	0	&	0	&	{0.990}	&	0	&	0	&	0.010	\\
		&	NB $k=8$	&	AIC  &	0	&	0	&	0	&	0.746	&	0	&	0	&	0.254	\\
		&		&	GIC  &	0	&	0	&	0	&	0.820	&	0	&	0	&	0.180\\
		&		&	BIC 	&	0	&	0	&	0	&	0.938	&	0	&	0	&	0.062	\\
		&		&	ELCIC &	0	&	0	&	0	&	{0.968}	&	0	&	0.002	&	0.030	\\
		&	NB $k=2$	&	AIC &	0	&	0	&	0	&	0.576	&	0	&	0	&	0.424	\\
		&		&	GIC 	&	0	&	0	&	0	&	0.772	&	0	&	0	&	0.228	\\
		&		&	BIC  	&	0	&	0	&	0	&	0.804	&	0	&	0	&	0.196	\\
		&		&	ELCIC &	0	&	0	&	0	&	{0.980}	&	0	&	0	&	0.020	\\
		\bottomrule         
	\end{tabular}
    \end{minipage}%
}
	\label{table1}
\end{table}

\textbf{Case 1}: \blue{The main goal in this case is to tell how distribution misspecification would affect the variable selection.} We narrow our focus to the Poisson regression, a special case of GLM. The true mean structure is 
\begin{equation*}
\log(\mu_{i})=\beta_0+\bfX_i^\T\bfbeta \qquad \text{for}\quad i = 1,\ldots,n,
\end{equation*}
where $\beta_0=0.5$, $\bfbeta=(0.5,0.5,0)^\T$. $\bfX_i$ is a covariate vector from a three-dimensional multivariate normal distribution $\textsc{MVN}(\bf0, \bfV)$, where the variance-covariance matrix $\bfV$ is an AR1 matrix with the variance as $1$ and the correlation coefficient as $0.5$. Also, we apply negative binomial as the true distribution with its number of failures $k=2,8$ to account for the case where the variance in the Poisson distribution is misspecified. However, we utilize AIC, GIC, and BIC for variable selection under the assumption of the Poisson distribution and apply ELCIC under the full estimating equations $\bfg$ specified in (\ref{poi}). The correct specification of the Poisson distribution is also considered as a benchmark. 500 Monte Carlo datasets with sample size $n = 100$, $200$, $400$ are generated, and the selection rates for each candidate model are reported for comparison. Based on Table \ref{table1}, we find out that when the variance structure is correctly specified, ELCIC is comparable to BIC, with a slightly worse performance, which is understandable since BIC incorporates all the likelihood information we need about the data. When the variance structure is misspecified, however, we observe that ELCIC reverses the situation and behaves much more robust than AIC and BIC as $n$ increases. More advantages can be gained by ELCIC when the data have higher over-dispersion. On the other hand, compared to the ELCIC, the consistency property does not hold for the AIC and GIC. Moreover, although GIC somehow relaxes the distribution assumption and performs more robust than AIC, it is still sensitive to the distribution misspecification compared with ELCIC even under a relatively large sample size. More discussions can be referred to Section \ref{discussion}.   

\begin{table}
	\centering
	\renewcommand\arraystretch{0.8}
		\caption{Performance of ELCIC compared with QIC for the scenarios under longitudinal count data. 500 Monte Carlo datasets are generated with sample size $n=100, 300$ and the number of observations within-subject $T=3, 5$. The model with $\{x_1,x_2\}$ and an exchangeable (EXC) correlation structure with the correlation coefficient $\rho=0.5$ is the true model. \redreview{AR1: auto-correlation 1; IND: independence; QIC/b: QIC with the BIC penalty.}}
	\begin{tabular}{lllcccccc}
		\toprule
		\multirow{2}{*}{Set-ups} & \multirow{2}{*}{Criteria} & \multirow{2}{*}{} & \multicolumn{6}{c}{Candidate Models} \\ 
		&   &  & $x_1,x_2,x_3$ & \boldsymbol{$x_1,x_2$} & $x_1,x_3$ & $x_2,x_3$ & $x_1$ & $x_3$  \\ \midrule
		$n=100$ &    ELCIC    &    {EXC}    &    0.040    &    {0.844}    &    0    &    0.002    &    0    &    0    \\
		$T=3$&        &    AR1    &    0.008    &    0.106    &    0    &    0    &    0    &    0    \\
		&        &    IND    &    0    &    0    &    0    &    0    &    0    &    0    \\
		&    QIC    &    {EXC}    &    0.090    &    {0.494}    &    0    &    0    &    0    &    0    \\
		&        &    AR1    &    0.044    &    0.372    &    0    &    0    &    0    &    0    \\
		&        &    IND    &    0    &    0    &    0    &    0    &    0    &    0    \\
		& QIC/b & {EXC} & 0.012 & {0.570} & 0  & 0.002 & 0  & 0  \\
 &       & AR1 & 0.018 & 0.398 & 0  & 0  & 0  & 0  \\
 &       & IND & 0  & 0  & 0 & 0  & 0  & 0 \\
		\midrule
		$n=300$ & ELCIC & {EXC} & 0.026 & {0.958} & 0 & 0 & 0 & 0 \\
		$T=3$ &  & AR1 & 0 & 0.016 & 0 & 0 & 0 & 0 \\
		&  & IND & 0 & 0 & 0 & 0 & 0 & 0 \\
		& QIC & {EXC} & 0.078 & {0.574} & 0 & 0 & 0 & 0 \\
		&  & AR1 & 0.030 & 0.318 & 0 & 0 & 0 & 0 \\
		&  & IND & 0 & 0 & 0 & 0 & 0 & 0 \\
		& QIC/b & {EXC} & 0.012 & {0.640} & 0 & 0 & 0 & 0 \\
 &       & AR1 & 0.002 & 0.346 & 0 & 0 & 0 & 0 \\
 &       & IND & 0     & 0     & 0 & 0 & 0 & 0\\\midrule
		$n=100$  &	ELCIC	&	{EXC}	&	0.052	&	{0.946}	&	0	&	0	&	0	&	0	\\
$T=5$ &		&	AR1	&	0	&	0.002	&	0	&	0	&	0	&	0	\\
&		&	IND	&	0	&	0	&	0	&	0	&	0	&	0	\\
&	QIC	&	{EXC}	&	0.102	&	{0.834}	&	0	&	0	&	0	&	0	\\
&		&	AR1	&	0.006	&	0.058	&	0	&	0	&	0	&	0	\\
&		&	IND	&	0	&	0	&	0	&	0	&	0	&	0	\\
& QIC/b & {EXC} & 0.016 & {0.920} & 0 & 0 & 0 & 0  \\
 &       & AR1 & 0  & 0.064 & 0 & 0  & 0  & 0  \\
 &       & IND & 0  & 0 & 0 & 0 & 0  & 0\\\midrule
$n=300$  &	ELCIC	&	{EXC}	&	0.02	&	{0.98}	&	0	&	0	&	0	&	0	\\
$T=5$ &		&	AR1	&	0	&	0	&	0	&	0	&	0	&	0	\\
&		&	IND	&	0	&	0	&	0	&	0	&	0	&	0	\\
&	QIC	&	{EXC}	&	0.098	&	{0.894}	&	0	&	0	&	0	&	0	\\
&		&	AR1	&	0.002	&	0.006	&	0	&	0	&	0	&	0	\\
&		&	IND	&	0	&	0	&	0	&	0	&	0	&	0	\\
 & QIC/b & {EXC} & 0.008 & {0.984} & 0 & 0 & 0 & 0 \\
 &       & AR1 & 0     & 0.008 & 0 & 0 & 0 & 0 \\
 &       & IND & 0     & 0     & 0 & 0 & 0 & 0\\
		\bottomrule
	\end{tabular}
		\label{table2}
	\end{table}

\textbf{Case 2}: We apply our proposed criterion to the GEE framework and make comparisons with one popular measure QIC. Suppose that the true underlying correlation structure is exchangeable (EXC) with the correlation coefficient $\rho = 0.5$. We assume count outcomes with the true marginal mean defined as
\begin{equation*}
\log(\mu_{ij})=\beta_0+x_{i1}\beta_1+x_{ij2} \beta_2 \qquad \text{for}\quad i = 1,\ldots,n, ~j = 1,\ldots, T,
\end{equation*}
where $x_{i1}$ is the subject (cluster) level covariate generated from an uniform distribution $\textsc{U}[0,1]$ and $x_{ij2} = j - 1$ is a time-dependent covariate. A redundant covariate $X_{ij3}$ is generated from a standard normal distribution $\textsc{N}(0,1)$. The number of observations (i.e., cluster size) is $T = 3$. The true parameters $\boldsymbol{\beta}=(-1,1,0.5)^\T$. 

As we discussed before, ELCIC with the full estimating equations (\ref{longfull}) could simultaneously select the marginal mean and correlation structures. However, QIC is not powerful enough to implement joint selection. Thus, we first apply the correlation information criteria (CIC) \citep{hin2009working} to identify the correlation structure under the full marginal mean structure and then implement QIC to select variables in marginal mean under the selected correlation structure. \redreview{We also compare the performance with QIC/b where the penalty of QIC is replaced by BIC penalty}. 500 Monte Carlo datasets with sample size $n=100$ or $300$ and observation times $T=3$ or $T=5$ are generated, and the selection rates for each combination of marginal mean and correlation structures are summarized in Table \ref{table2}. It is easy to recognize that the two-stage selection procedure based on CIC and QIC \redreview{(QIC/b)} are less powerful, in particular for $T=3$, since the second stage for variable selection heavily relies on the first stage for correlation structure selection. On the other hand, our ELCIC still keeps a much higher selection rate across various scenarios, which shows more robustness for model selection in longitudinal data framework. More detailed discussion can be seen in Section \ref{discussion}.   

To see the flexibility of ELCIC to handle nuisance parameters, we also implement the variable selection by using the first part estimating equations in (\ref{longfull}) as our $\bfg(\bfD_i,\bfgamma)$, thus regarding correlation coefficients as nuisance parameters. We observe a higher selection rate performance compared to QIC. The results are displayed in the Supplementary Material. 

\begin{table}
	\centering
	\caption{Performance of ELCIC compared with cross-validation (CV) method for the scenarios under longitudinal count data. 500 Monte Carlo data are generated with sample size $n=100, 200$ and the number of observations within-subject $T=3$.}
	\begin{threeparttable}
		\begin{tabular}{p{1.6cm}p{1.3cm} p{0.9cm}  p{0.9cm} p{0.9cm} p{0.9cm} p{0.9cm} p{0.9cm} p{0.9cm}}
			\toprule
			Sample Size & Criteria  & MS & FP & OVS & OVOS & CS & JS & VOS  \\\midrule
			$n=100$ & CV & 0.017 & 1.710 & 0.225 & 0.775 &  --  &   --   &     --  \\
			& ELCIC$_1$     & 0.016 & 0.595 & 0.575 & 0.425 &    --   &  --    &--      \\
			& ELCIC$_2$   & 0.016 & 0.610 & 0.565 & 0.435 & 0.91  & 0.52 & 0.39  \\ \midrule
			$n=200$ & CV & 0.007 & 1.825 & 0.190 & 0.810 &   --    & --     &    --   \\
			& ELCIC$_1$     & 0.007 & 0.405 & 0.655 & 0.345 &    --   &  --    &  --     \\
			& ELCIC$_2$   & 0.007 & 0.405 & 0.655 & 0.345 & 0.975 & 0.64 & 0.335\\  \bottomrule
		\end{tabular}
	\end{threeparttable}
	\label{table3}
	\begin{tablenotes}
		\item Note: MS: consistency $\lVert\hat{\bfbeta}-\bfbeta_0\lVert^2$; FP: average number of falsely selecting the non-zero variables; OVS: number of selecting the true mean structure/500; OVOS: number of over selecting the mean structure/500; CS: number of selecting the true correlation structure/500; JS: number of joint selecting true mean and correlation structures/500; VOS: number of over selecting the mean structure under true correlation structure selected/number of selecting the true correlation structure.
	\end{tablenotes}
\end{table}

\textbf{Case 3}: We simulate data from the model $\bfY_i=\bfX_i^T\bfbeta+\bfepsilon_i$, $i=1,\ldots,n$ where $\bfbeta=(0.5,0.5,0.5,0,0,0,0)^\T$. Note $\bfX_i$ is a $T\times 7$ matrix from a multi-variate normal distribution $\textsc{MVN}(\bf0, \bfV)$, where $T=3$, and the variance covariance matrix $\bfV$ is an AR1 matrix with the variance as 1 and the correlation coefficient as $0.5$. The random errors $\{\bfepsilon_i\}$'s are also generated from a multi-variate normal distribution, with zero-mean and an exchangeable covariance matrix with $\sigma^2=1$ and $\rho=0.5$. 500 Monte Carlo datasets with sample size $n=100$ or $200$ are generated, and we record several important measures such as consistency, false positive, overall variable selection rate, overall variable over-selection rate, correlation structure selection rate, joint selection rate, and variable over-selection rate given the true correlation structure is selected. Both ELCIC$_1$ and ELCIC$_2$ utilize the full estimating equations (\ref{longfull}). Note that ELCIC$_1$ is calculated under the true correlation structure $\bfrho^{EXC}$ to compare with CV for variable selection, but ELCIC$_2$ is for jointly selecting marginal mean and correlation structures. In Table \ref{table3}, we find that ELCIC$_1$ tends to have substantial lower false positive and higher variable selection rate compared to CV-based method, and ELCIC$_2$ also has satisfactory performance for simultaneous selection where CV cannot be applicable. 

\section{Real Data Example}\label{section4}
We apply our method to the Atherosclerosis Risk in Communities Study (ARIC), originally designed to investigate the causes of atherosclerosis and its clinical outcomes, the trends in rates of hospitalized myocardial infarction and coronary heart diseases. The platelet counts are our outcomes, which have been studied in the literature and realized to be an essential factor in coronary heart diseases \citep{renaud1992wine}. The objective of this application is to investigate the temporal pattern of platelet count and identify potential risk factors among various baseline variables, such as age (year), gender (female/male), diabetes (1=yes/0=no), smoker (1=yes/0=no), body mass index (kg), total cholesterol (mmol/L), total triglycerides (mmol/L), and time-dependent variable visit (coded as 0,1,2,3). We select Washington County to identify a total of 1,463 white patients at approximately three-year intervals (1987-1989, 1990-1992, 1993-1995, and 1996-1998) who were diagnosed with hypertension at the first examination. Note that there were 441 dropouts during the follow-up. To illustrate the application of our proposal, we assume missing completely at random for simplicity; however, more sophisticated manipulation of missing data can be referred to \cite{chen2018empirical}. We apply the full estimating equations (\ref{longfull}) with GEE in Case 2 to construct ELCIC and facilitate joint selection of marginal mean and correlation structures. The results are summarized in Table \ref{table4}, indicating that the marginal mean including time, gender, age, diabetes, and cholesterol with the AR1 correlation structure is recommended as the optimal model by both QIC and ELCIC. Note that the variables selected by ELCIC match the significant ones when we fit the full model (Model 1). The same marginal mean model is identified by the PGEE procedure from Case 3 using the same covariate pool.

\begin{table}
	\centering
	\caption{Analysis of the ARIC study based on six candidate marginal mean and potential correlation structures \redreview{EXC: exchangeability; AR1: auto-correlation 1; IND: independence}}
	\scalebox{0.88}{%
	\begin{minipage}[c]{1.2\textwidth}
		\begin{tabular}{p{1.2cm}p{0.5cm} p{1.8cm} p{1.8cm}p{1.8cm}p{1.8cm}p{1.8cm}p{1.8cm}}
			\toprule
			Variables & & Model 1 & Model 2 & Model 3  & Model 4 & Model 5  & Model 6 \\\midrule
			Time &  & -0.025 (0.002)* & -0.025 (0.002)* & -0.025 (0.002)* & -0.026 (0.002)* & -0.026 (0.002)* & -0.025 (0.002)* \\
			Gender && 0.167 (0.014)*  & 0.174 (0.013)*  & 0.171 (0.013)*   & 0.171 (0.013)*   & 0.162 (0.013)*  & 0.167 (0.014)*  \\
			Smoke&  & 0.017 (0.014)     &   &   &  &  & 0.016 (0.014)     \\
			Age(year) &   & -0.032 (0.012)*  &  & -0.029 (0.012)*   & -0.026 (0.012)*   & -0.033 (0.012)*  & -0.033 (0.012)*  \\
			Diabetes &   & -0.058 (0.020)*  &   &   & -0.054 (0.019)*  & -0.055 (0.019)*  & -0.056 (0.019)*  \\
			BMI &   & 0.001 (0.001)     &   &      &   &   &   \\
			Cholesterol &   & 0.024 (0.006)*  &    &   &    & 0.025 (0.006)* & 0.024 (0.006)*  \\
			Triglycerides&    & 0.002 (0.005)  &  &  &  &  & 0.002 (0.005)\\ 
			\midrule
			\textbf{ELCIC} &EXC& 91.8&    81.6&    82.9&    80.3&    69.5&    84.4\\
			&  AR1&     86.7&    84.5&    85.4&    85.4&    \textbf{68.4}&    80.6\\
			&  IND&     963.1&    893.4&    949.4&    959.1&    956&    948.8\\
			\textbf{QIC}& AR1&	-183571.2&	-183485.4&	-183508.6&	-183529&	\textbf{-183572}&	-183571.6\\ 
			\bottomrule
		\end{tabular}
    \end{minipage}%
}
	\label{table4}
	\begin{tablenotes}
		\item \redreview{Note that estimations with standard errors (in parenthesis) are reported; the criteria of ELCIC and QIC are summarized for model selection;	*denotes the P-value$<0.05$.}
	\end{tablenotes}
\end{table}

\section{Discussion}\label{discussion}

We are in the midst of big data, and increasing amounts of such data have been available in many scientific fields. Due to complex features and data structure, the existing traditional approaches for analysis and inference based on the specification of distribution may not work well. In this work, we present a data-driven information criterion framework for model selection under different contexts. By further relaxing the estimation procedure, our ELCIC overcomes the limitation of classic empirical likelihood-based criteria, \redreview{and more importantly, it is versatile for extension depending on model selection needs in practice, where no existing information criteria would theoretically fit well, such as the Case 3 and the extra cases shown in the Supplementary Material. Noted that, in the main context, we assume that the full model is correctly specified. however, by further checking Condition 1 and Condition 4, we want to highlight that the identifiability of ELCIC indeed does not require the correct specification of the full model. We only need the existence of the parameters which will make the full estimating equations equal to zeros. Consequentially, when the full model we utilize is misspecified, ELCIC will still work but locating the "true" one which leads the current estimating equations equal to zeros. \chixiang{However, the existence of such a "true" value is not theoretically trivial and also not always guaranteed, which may need more investigation in future work.} Thus, in practice, we recommend users to input all potentially important variables so that the selected full model could have good fit to the data, if enough sample size is provided. }

\blue{Through extensive numerical studies, the out-performance of our ELCIC can be demonstrated compared to the other alternatives.} GIC tends to be somewhat sensitive to the distribution misspecification, even under a large sample size. The reason might be two-folded: 1) a relatively complicated bias-correction term should be estimated based upon data structure and candidate models, thus leading to more variability to the selection; 2) The underlying measurement of GIC is defined as the Kullback-Leibler distance evaluated at the misspecified distribution, which could introduce systematic discrepancy to the truth, no matter how accurate the bias correction is. Therefore, GIC would intrinsically lose some power to capture the true model particularly when the specified distribution deviates farther from the truth. The same issue is also applied for QIC, regarding its misspecification of quasi-likelihood by considering the independent correlation structure \citep{pan2001akaike}.

Several extensions are open to research. \redreview{One is its extension to (ultra-)high dimensional cases. Per reviewers' suggestion, we looked into this framework by literature search and performing empirical simulation studies. Note that there are only a few papers related to the high dimensional empirical likelihood variable selection by \cite{tang2010} and \cite{chang2018}. In particular, \cite{tang2010} proposed the LASSO type of penalized empirical likelihood with diverging numbers of parameters in linear models. Later, \cite{chang2018} applied extra penalty in penalized empirical likelihood regulating the Lagrange multipliers so that ultra-high dimensional variable selection is theoretically applicable. To indicate the feasibility and potentials of this extension, we conducted simulation studies and found out out-performance of our proposal compared with other alternative criteria, which are provided in the Supplementary Material. However, the extended proofs for high-dimensional cases with empirical likelihood are not trivial and need more effort, which will be fully pursued in future studies.} \chixiang{Furthermore, the robust model prediction criterion would also be of independent research interest in the future, and informative degree of freedom instead of $p$ might be used to borrow more information, though it would increase the criteria complexity and case specificity and hence reduce its robustness in general practice.} As a starting point, ELCIC shows its potential and robustness to explore statistical issues related to model selection with highly complex and diverse data.

\vskip 14pt
\noindent {\large\bf Supplementary Material}

\redreview{The supplementary materials provide the detailed proofs for Theorems \ref{thm1}-\ref{thm4} and Corollary \ref{corollary}, and also the results for additional simulation studies including the cases for variable selection under the augmented inverse probability weighted estimator framework and the ultra-high dimensional set-up.}
\par
\vskip 14pt
\noindent {\large\bf Acknowledgments}

\redreview{The authors thank the editor, the associate editor and the referees for valuable suggestions.} Wang's research was partially supported by Grant UL1 TR002014 and KL2 TR002015 from the National Center for Advancing Translational Sciences (NCATS) and was also funded, in part, under a grant with the Pennsylvania Department of Health using Tobacco CURE Funds. The content is solely the responsibility of the authors and does not represent the official views of the National Institute of Health, and the Department specially disclaims responsibility for any analyses, interpretations or conclusions.
\par
\markboth{\hfill{\footnotesize\rm Chixiang Chen, Ming Wang, Rongling Wu and Runze Li} \hfill}
{\hfill {\footnotesize\rm  A Robust Consistent Information Criterion based on Empirical Likelihood} \hfill}

\bibhang=1.7pc
\bibsep=2pt
\fontsize{9}{14pt plus.8pt minus .6pt}\selectfont
\renewcommand\bibname{\large \bf References}
\expandafter\ifx\csname
natexlab\endcsname\relax\def\natexlab#1{#1}\fi
\expandafter\ifx\csname url\endcsname\relax
\def\url#1{\texttt{#1}}\fi
\expandafter\ifx\csname urlprefix\endcsname\relax\def\urlprefix{URL}\fi

\bibliographystyle{chicago}

\bibliography{Bibliography-MM-MC}


\end{document}



\renewcommand{\baselinestretch}{2}

\markright{ \hbox{\footnotesize\rm  Supplementary Material
}\hfill\\[-25pt]
\hbox{\footnotesize\rm
}\hfill }

\markboth{\hfill{\footnotesize\rm Chen et al.} \hfill}
{\hfill {\footnotesize\rm ELCIC for Model Selection} \hfill}

\renewcommand{\thefootnote}{}
$\ $\par \fontsize{12}{14pt plus.8pt minus .6pt}\selectfont


\centerline{\large\bf  Supplementary Material for ``A Robust Consistent Information Criterion for }
\vspace{2pt} 
\centerline{\large\bf  Model Selection based on Empirical Likelihood" by}
\vspace{.25cm}
 \centerline{Chixiang Chen$^1$, Ming Wang$^{1}$, Rongling Wu$^1$, Runze Li$^2$}
\vspace{.4cm}
\centerline{\it
1. Division of Biostatistics and Bioinformatics, Department of Public Health Science,}\centerline{\it Pennsylvania State College of Medicine, Hershey, PA, U.S.A., 17033}\centerline{\it
2. Department of Statistics and the Methodology Center, }\centerline{\it Pennsylvania State University, University Park, PA, U.S.A, 16802}
\vspace{.8cm}
\fontsize{9}{11.5pt plus.8pt minus .6pt}\selectfont
\noindent
\par

\setcounter{section}{0}
\setcounter{equation}{0}
\def\theequation{S\arabic{section}.\arabic{equation}}
\def\thesection{S\arabic{section}}

\fontsize{12}{14pt plus.8pt minus .6pt}\selectfont


\lhead[\footnotesize\thepage\fancyplain{}\leftmark]{}\rhead[]{\fancyplain{}\rightmark\footnotesize\thepage}
In the Supplementary Material, we will provide the technical proofs to the Theorems 1-4 and Corollary 1 in the main paper as well as the results for additional simulation studies. \red{Note that all expectations are evaluated at the true values. \chixiang{For convenience, we will utilize the following simplified notations: $E(\bfg\bfg^\T)=E(\bfg(\bfD,\bfgamma_0)\bfg^\T(\bfD,\bfgamma_0))$, $E(\partial\bfg/\partial\bfgamma^\T)=E(\partial\bfg(\bfD,\bfgamma_0)/\partial\bfgamma^\T)$;\\ $E(\bfg_1\bfg_1^\T)=E(\bfg_1(\bfD,\bfgamma_0)\bfg_1^\T(\bfD,\bfgamma_0))$, $E(\partial\bfg_1/\partial\bfgamma^\T)=E(\partial\bfg_1(\bfD,\bfgamma_0)/\partial\bfgamma^\T)$;\\
$E(\bfg_2\bfg_2^\T)=E(\bfg_2(\bfD,\bfgamma_0)\bfg_2^\T(\bfD,\bfgamma_0))$, $E(\partial\bfg_2/\partial\bfgamma^\T)=E(\partial\bfg_2(\bfD,\bfgamma_0)/\partial\bfgamma^\T)$;\\
$E(\bfg_1\bfg_2^\T)=E(\bfg_2\bfg_1^\T)^\T=E(\bfg_1(\bfD,\bfgamma_0)\bfg_2^\T(\bfD,\bfgamma_0))$.
}
} 

\section{Technical proofs}
\subsection{Proof of Theorem 1}

In order to show Theorem 1 in the main paper, we need the following lemma:
\begin{lemma}\label{lemma1}
   Under \textit{Condition} 1 and \textit{Condition} 2 in the main paper, we have the following relationship:
	\begin{equation*}
	\begin{split}
	\hat{\bfgamma}_{EL}=\hat{\bfgamma}_{EE}+o_p(\bfn^{-\frac{1}{2}}),\\
	\hat{\bflambda}_{EL}=\hat{\bflambda}_{EE}+o_p(\bfn^{-\frac{1}{2}}).
	\end{split}
	\end{equation*}
\end{lemma}

\begin{proof}
	Let $\tilde{\bfgamma}_{0}=(\bfgamma^\T_{0},\mathbf{0}^\T)^\T$. Along the lines with the proof of Lemma 1 and Theorem 1 in \cite{qin1994} under \textit{Condition 1}, \red{and based on asymptotic theory in generalized method of moment \citep{newey1994large}, we have}
	\begin{equation}\label{lm11}
	\begin{split}
	&\hat{\bfgamma}_{EL}-\bfgamma_0=-\bfGamma^{-1} E\Big(\frac{\partial\bfg}{\partial\bfgamma^\T}\Big)^\T \Big(E\bfg\bfg^\T\Big)^{-1} \frac{1}{n} \sum_{i=1}^{n}\bfg(\bfD_i,\bfgamma_{0})+o_p(\bfn^{-\frac{1}{2}}),\\
	&\hat{\bfgamma}_{EE}-\bfgamma_0=-E\Big(\frac{\partial\bfg_1}{\partial\bfgamma^\T}\Big)^{-1} \frac{1}{n} \red{\sum_{i=1}^{n}\bfg_1(\bfD_i,\tilde{\bfgamma}_{0})+o_p(\bfn^{-\frac{1}{2}})},
	\end{split}
	\end{equation}
	with $\bfGamma= E(\partial\bfg/\partial\bfgamma^\T)^\T(E\bfg\bfg^\T)^{-1}E(\partial\bfg/\partial\bfgamma^\T)$. \red{Note that $\bfg_1(\bfD_i,\tilde{\bfgamma}_{0})=\bfg_1(\bfD_i,{\bfgamma}_{0})$. So we can keep using $\bfg_1(\bfD_i,{\bfgamma}_{0})$ in the following proof.}
	Thus, we have:
	\begin{equation*}
	\begin{split}
	\bfGamma
	=&\left(E\Big(\frac{\partial\bfg_1}{\partial\bfgamma^\T}\Big)^\T, -E\Big(\frac{\partial\bfg_1}{\partial\bfgamma^\T}\Big)^\T(E\bfg_1\bfg_1^\T)^{-1}E(\bfg_1\bfg_2^\T)+E\Big(\frac{\partial\bfg_2}{\partial\bfgamma^\T}\Big)^\T\right)\times\\
	& \left(\begin{matrix}
	E(\bfg_1\bfg_1^\T)^{-1} & 0\\ 0 & \bfA
	\end{matrix}\right)\left(\begin{matrix}
	E (\frac{\partial\bfg_1}{\partial\bfgamma^\T} )\\ 
	-E(\bfg_2\bfg_1^\T)E(\bfg_1\bfg_1^\T)^{-1} E(\frac{\partial\bfg_1}{\partial\bfgamma^\T} ) +E (\frac{\partial\bfg_2}{\partial\bfgamma^\T} ),
	\end{matrix}\right),\\
	\end{split}
	\end{equation*}
	with \red{$\bfA=E(\bfg_2\bfg_2^\T)-E(\bfg_2\bfg_1^\T)(E\bfg_1\bfg_1^\T)^{-1}E(\bfg_1\bfg_2^\T)$}. By \textit{Condition} 2, we have $E(\partial\bfg_2/\partial\bfgamma^\T)^\T=E(\partial\bfg_1/\partial\bfgamma^\T)^\T (E\bfg_1 \bfg_1^\T)^{-1} E( \bfg_1 \bfg_2^\T)$, \red{which leads to}
	\begin{equation}\label{lm12}
	\bfGamma=E\Big(\frac{\partial\bfg}{\partial\bfgamma^\T}\Big)^\T (E \bfg \bfg^\T)^{-1} E\Big(\frac{\partial\bfg}{\partial\bfgamma^\T}\Big)= E\Big(\frac{\partial\bfg_1}{\partial\bfgamma^\T}\Big)^\T (E \bfg_1 \bfg_1^\T)^{-1} E\Big(\frac{\partial\bfg_1}{\partial\bfgamma^\T}\Big).
	\end{equation}
Similarly, applying \textit{Condition} 2 again, we have
	\begin{equation}\label{lm13}
	E\Big(\frac{\partial\bfg}{\partial\bfgamma^\T}\Big)^\T (E\bfg \bfg^\T)^{-1}  \frac{1}{n} \sum_{i=1}^{n} \bfg(\bfD_i,\bfgamma_{0}) =  E\Big(\frac{\partial\bfg_1}{\partial\bfgamma^\T}\Big)^\T (E \bfg_1 \bfg_1^\T)^{-1} \frac{1}{n} \sum_{i=1}^{n} \red{\bfg_1(\bfD_i,{\bfgamma}_{0})}.
	\end{equation}
Substituting (\ref{lm12}) and (\ref{lm13}) into (\ref{lm11}), and by noticing that $E(\partial\bfg_1/\partial\bfgamma^\T)$ is invertible, we have
	\begin{equation*}
	\begin{split}
	\hat{\bfgamma}_{EL}-\bfgamma_0=&-E\Big(\frac{\partial\bfg_1}{\partial\bfgamma^\T}\Big)^{-1}\frac{1}{n} \sum_{i=1}^{n} \red{\bfg_1(\bfD_i,{\bfgamma}_{0})} +o_p(\bfn^{-\frac{1}{2}})\\
	=&\hat{\bfgamma}_{EE}-\bfgamma_0+o_p(\bfn^{-\frac{1}{2}}).
	\end{split}
	\end{equation*}
	
	Finally, by Taylor expansion of the equation (2.4) in the main paper at $\bfgamma_0$ and $\bflambda_0=\mathbf{0}$, \red{we can derive the following equation} \begin{equation*}
	\frac{1}{n} \sum_{i=1}^{n}\bfg(\bfD_i,\bfgamma_{0})+E\Big(\frac{\partial\bfg}{\partial\bfgamma^\T}\Big)(\hat{\bfgamma}_{EE}-\bfgamma_0)-E(\bfg\bfg^\T)\hat{\bflambda}_{EE}+o_p(\bfn^{-\frac{1}{2}})=\mathbf{0},
	\end{equation*}
	leading to $\hat{\bflambda}_{EL}=\hat{\bflambda}_{EE}+o_p(\bfn^{-1/2})$ by applying $\hat{\bfgamma}_{EL}=\hat{\bfgamma}_{EE}+o_p(\bfn^{-1/2})$.
\end{proof}	

Now we are ready to show Theorem 1, and the proof is shown next. 
\begin{proof}	
	\red{Now consider any random variables $\bflambda$ and $\bfgamma$ satisfying $\bflambda-\bflambda_0=O_P(\bfn^{-1/2})$} and $\bfgamma-\bfgamma_0=O_P(\bfn^{-1/2})$ and expand $\log R^F(\bflambda,\bfgamma)$ at $\bflambda^*$ and $\bfgamma^*$ such that $\log R^F(\bflambda^*,\bfgamma^*)$ is maximized. Indeed,	these two $\bflambda^*$ and $\bfgamma^*$ are the maximum empirical likelihood estimates $\hat{\bflambda}_{EL}$ and $\hat{\bfgamma}_{EL}$, respectively; given $l=-\log R^F(\bflambda,\bfgamma)$, we have\red{
	\begin{equation*}
	\begin{split}
	\log R^F(\bflambda,\bfgamma)&=\log R^F(\hat{\bflambda}_{EL},\hat{\bfgamma}_{EL})-\frac{\partial l}{\partial \bflambda}\Big|_{\bflambda= \hat{\bflambda}_{EL},\bfgamma=\hat{\bfgamma}_{EL}} (\bflambda-\hat{\bflambda}_{EL})\\
	&-\frac{\partial l}{\partial \bfgamma}\Big|_{\bflambda= \hat{\bflambda}_{EL},\bfgamma=\hat{\bfgamma}_{EL}} (\bfgamma-\hat{\bfgamma}_{EL}) \\
	&- \frac{1}{2}(\bflambda-\hat{\bflambda}_{EL})^\T \frac{\partial^2 l}{\partial \bflambda\partial \bflambda^\T}\Big|_{\bflambda= \hat{\bflambda}_{EL},\bfgamma=\hat{\bfgamma}_{EL}} (\bflambda-\hat{\bflambda}_{EL}) \\
	&-(\bflambda-\hat{\bflambda}_{EL})^\T \frac{\partial^2 l}{\partial \bflambda\partial \bfgamma^\T}\Big|_{\bflambda= \hat{\bflambda}_{EL},\bfgamma=\hat{\bfgamma}_{EL}} (\bfgamma-\hat{\bfgamma}_{EL})\\
	&-\frac{1}{2} (\bfgamma-\hat{\bfgamma}_{EL})^\T \frac{\partial^2 l}{\partial \bfgamma\partial \bfgamma^\T}\Big|_{\bflambda= \hat{\bflambda}_{EL},\bfgamma=\hat{\bfgamma}_{EL}} (\bfgamma-\hat{\bfgamma}_{EL})+o_p(1).
	\end{split}
	\end{equation*}}
	To be noted that $(\partial l/\partial \bflambda)\mid_{\bflambda= \hat{\bflambda}_{EL},\bfgamma=\hat{\bfgamma}_{EL}}=(\partial l/\partial \bfgamma)\mid_{\bflambda= \hat{\bflambda}_{EL},\bfgamma=\hat{\bfgamma}_{EL}} =\mathbf{0}$ by definition of $\hat{\bflambda}_{EL}$ and $\hat{\bfgamma}_{EL}$. Furthermore, applying the weak law of large number, we have $(1/n)\partial^2 l/(\partial \bfgamma\partial \bfgamma^\T)\xrightarrow{P} \mathbf{0}$, $(1/n) \partial^2 l/(\partial \bflambda\partial \bflambda^\T)\xrightarrow{P} \red{-E \{\bfg(\bfD, \bfgamma_0)\bfg^\T(\bfD, \bfgamma_0) \}}$, and $(1/n) \partial^2 l/(\partial \bflambda\partial \bfgamma^\T)\xrightarrow{P} \red{E\{\partial\bfg(\bfD, \bfgamma_0)/\partial\bfgamma^\T\}}$. \red{By utilizing all the derived results, we rewrite the logarithm of empirical likelihood ratio as} \red{
	\begin{equation*}
	\begin{split}
	\log R^F(\bflambda,\bfgamma)= &-\frac{1}{2}\Big((\bflambda-\hat{\bflambda}_{EL})^\T,(\bfgamma-\hat{\bfgamma}_{EL})^\T\Big)^\T \left(\begin{matrix}
	-n\bfSigma_{11}&n\bfSigma_{12}\\ n\bfSigma_{21}&\mathbf{0}
	\end{matrix}\right)\left(\begin{matrix}
	\bflambda-\hat{\bflambda}_{EL}\\ \bfgamma-\hat{\bfgamma}_{EL} 
	\end{matrix}\right)\\
	&+\log R^F(\hat{\bflambda}_{EL},\hat{\bfgamma}_{EL})+o_p(1),
	\end{split}
	\end{equation*}}
	where \red{$\bfSigma_{11}=E\{\bfg(\bfD, \bfgamma_0)\bfg^\T(\bfD, \bfgamma_0)\}$, $\bfSigma_{12} =E\{\partial\bfg(\bfD, \bfgamma_0)/\partial\bfgamma^\T\}$}, $\bfSigma_{21} =\bfSigma_{12}^\T$. \red{By Lemma \ref{lemma1} and some algebra, we can have $\log R^F(\bflambda,\bfgamma)=\log R^F(\hat{\bflambda}_{EE},\hat{\bfgamma}_{EE})-(1/2)\bfdelta^\T(n\bfSigma)\bfdelta+o_p(1)$ with
	\begin{equation*}
	\bfSigma=\left(\begin{matrix}
	-\bfSigma_{11}&\mathbf{0}\\ \mathbf{0}&\bfSigma_{21}\bfSigma^{-1}_{11}\bfSigma_{12}
	\end{matrix}\right) \text{and}~ \bfdelta= \left(\begin{matrix} \bfdelta_1 \\ \bfdelta_2
	\end{matrix}\right),
	\end{equation*}
where $\bfdelta_1=\bflambda-\hat{\bflambda}_{EE}+\bfSigma_{21}\bfSigma_{11}^{-1}(\bfgamma-\hat{\bfgamma}_{EE})$ and $\bfdelta_2=\bfgamma-\hat{\bfgamma}_{EE}$.}	
	
	\red{In the end, by integrating out the random variables $\bfdelta$, the marginal probability will become
	\begin{equation*}
	\begin{split}
	P(\bfY| M) =&  R^F(\hat{\bflambda}_{EE},\hat{\bfgamma}_{EE}) \int\exp \Big\{-\frac{1}{2}\bfdelta^\T(n\bfSigma)\bfdelta\Big\}\rho_{\bfdelta}(\bfdelta)d\bfdelta+o_p(1)\\
	=&  R^F(\hat{\bflambda}_{EE},\hat{\bfgamma}_{EE}) \int\exp \Big\{\frac{1}{2}\bfdelta_1^\T(n\bfSigma_{11})\bfdelta_1\Big\}\rho_{\bfdelta_1}(\bfdelta_1)d\bfdelta_1\\
	&\cdot\int\exp \Big\{-\frac{1}{2}\bfdelta_2^\T(n\bfSigma_{21}\bfSigma^{-1}_{11}\bfSigma_{12})\bfdelta_2\Big\}\rho_{\bfdelta_2}(\bfdelta_2)d\bfdelta_2+o_p(1)\\
	\end{split}
	\end{equation*}
	Thus, by applying non-informative prior to $\bfdelta_2$, i.e., $\rho_{\bfdelta_2}(\bfdelta_2)=1$, and applying the Laplace approximation, we can have $P(\bfY| M) =R^F(\hat{\bflambda}_{EE},\hat{\bfgamma}_{EE}) (2\pi)^{p/2}\big|n\bfSigma_{21}\bfSigma^{-1}_{11}\bfSigma_{12}\big|^{-1/2}\tilde{A}+o_p(1)$, with $\tilde{A}=\int\exp \Big\{\frac{1}{2}\bfdelta_1^\T(n\bfSigma_{11})\bfdelta_1\Big\}\rho_{\bfdelta_1}(\bfdelta_1)d\bfdelta_1$, which finally leads to the conclusion by taking negative two logarithm of marginal probability: $-2\log P(\bfY| M) = -2\log R^F(\hat{\bflambda}_{EE},\hat{\bfgamma}_{EE}) +p\log n+\log|\bfSigma_{21}\bfSigma^{-1}_{11}\bfSigma_{12}|-p\log(2\pi)-2\log(\tilde{A})+o_p(1)$.}
\end{proof}	

\subsection{Proof of Theorem 2}
\begin{proof}
	
	Given \textit{Condition} 1 and along the lines of the proofs in \cite{owen2001}, for any $\bfgamma$ satisfying $\lVert \bfgamma-\bfgamma_{0}\rVert\leq Cn^{-1/3}$ with a large enough constant $C>0$, we have $\hat{\bflambda}(\boldsymbol{\gamma})=o_p(\bfn^{-1/3})$ and $\max_{1\leq i\leq n}\lVert \bfg(\bfD_i, \bfgamma)\rVert=o_p(n^{-1/3})$. Accordingly, $\max_{1\leq i\leq n} \hat{\bflambda}^\T(\boldsymbol{\bfgamma})\bfg(\bfD_i, \bfgamma)=o_p(1)$ uniformly for $\lVert\bfgamma-\bfgamma_{0}\rVert\leq Cn^{-1/3}$. On the other hand, by the definition of $\hat{\bflambda}_{EE}$ and taking first order Taylor expansion at $\bfgamma_0$ and $\bflambda_0=\mathbf{0}$, we have the following equation:
	\begin{equation*}
	\begin{split}
	\mathbf{0}=&\frac{1}{n}\frac{\partial l}{\partial \bflambda}\Big|_{\bflambda= \hat{\bflambda}_{EE},\bfgamma=\hat{\bfgamma}_{EE}} 
	=\frac{1}{n} \sum_{i=1}^n \frac{\bfg(\bfD_i, \hat{\bfgamma}_{EE})}{1+\hat{\bflambda}_{EE}\bfg(\bfD_i, \hat{\bfgamma}_{EE})}\\
	=&\frac{1}{n} \sum_{i=1}^n \bfg(\bfD_i,  {\bfgamma}_{0})+\frac{1}{n} \sum_{i=1}^n \red{\frac{\bfg(\bfD_i,  \bfgamma_0 )}{\partial\bfgamma^\T}} (\hat{\bfgamma}_{EE}-\bfgamma_{0})\\
	&-\frac{1}{n} \sum_{i=1}^n \bfg(\bfD_i,  {\bfgamma}_{0})\bfg^\T(\bfD_i,  {\bfgamma}_{0}) (\hat{\bflambda}_{EE}-\mathbf{0})+o_p(\bfvarepsilon_n),
	\end{split}
	\end{equation*}
	where \red{$\bfvarepsilon_n=\lVert\hat{\bfgamma}_{EE}-\bfgamma_0\lVert+\lVert\hat{\bflambda}_{EE}\lVert$}. By solving $\hat{\bflambda}_{EE}$ from the above formula, we have
	\begin{equation}\label{thm21}
	\hat{\bflambda}_{EE}=\bfS_n^{-1}\Big\{\frac{1}{n} \sum_{i=1}^n \bfg(\bfD_i,  {\bfgamma}_{0}) +\frac{1}{n} \sum_{i=1}^n \red{\frac{\bfg(\bfD_i,  \bfgamma_0 )}{\partial\bfgamma^\T}}(\hat{\bfgamma}_{EE}-\bfgamma_{0}) +o_p(\bfvarepsilon_n) \Big\},
	\end{equation}
	where $\bfS_n=(1/n) \sum_{i=1}^n \bfg(\bfD_i,  {\bfgamma}_{0}) \bfg^\T(\bfD_i,  {\bfgamma}_{0})$. Also, $\hat{\bfgamma}_{EE}-\bfgamma_{0}=O_p(\bfn^{-1/2})$ and $(1/n) \sum_{i=1}^n \bfg(\bfD_i, {\bfgamma}_{0})=O_p(\bfn^{-1/2})$, we conclude that $\bfvarepsilon_n=O_p(\bfn^{-1/2})$. Thus, by \textit{Condition} 1 and the weak law of large number, $\hat{\bflambda}_{EE}$ is rewritten as
	\begin{equation}\label{thm22}
	\hat{\bflambda}_{EE}=\bfSigma_{11}^{-1}\Big\{\frac{1}{n} \sum_{i=1}^n \bfg(\bfD_i,  {\bfgamma}_{0}) + \bfSigma_{12}\big(\hat{\bfgamma}_{EE} -\bfgamma_{0}\big)\Big\}+ o_p(\bfn^{-\frac{1}{2}}).
	\end{equation}
	
	Now let us expand $l$ in the following manner:
	\begin{equation*}
	\begin{split}
	l=&\sum_{i=1}^n\log\Big\{1+\hat{\bflambda}_{EE} \bfg(\bfD_i,  \hat{\bfgamma}_{EE})\Big\}\\
	=&\hat{\bflambda}_{EE}^\T\sum_{i=1}^n\bfg(\bfD_i,  \hat{\bfgamma}_{EE})-\frac{1}{2} \hat{\bflambda}_{EE}^\T\Big\{ \sum_{i=1}^n\bfg(\bfD_i,  \hat{\bfgamma}_{EE}) \bfg^\T(\bfD_i,  \hat{\bfgamma}_{EE})\Big\} \hat{\bflambda}_{EE}^\T\\
	&+O_p\Big(\sum_{i=1}^n\big\{\hat{\bflambda}_{EE}^\T \bfg(\bfD_i,  \hat{\bfgamma}_{EE})\big\}^3\Big).
	\end{split}
	\end{equation*}

It is noted that
	$$\sum_{i=1}^n\big\{\hat{\bflambda}_{EE}^\T \bfg(\bfD_i,  \hat{\bfgamma}_{EE})\big\}^3 \leq \sum_{i=1}^n\big\{\hat{\bflambda}_{EE}^\T \bfg(\bfD_i,  \hat{\bfgamma}_{EE})\big\}^2 \max_{1\leq i\leq n} \hat{\bflambda}_{EE}^\T \bfg(\bfD_i,  \hat{\bfgamma}_{EE}).$$
	\red{Thus, by realizing that} $\max_{1\leq i\leq n} \hat{\bflambda}_{EE}^\T \bfg(\bfD_i,  \hat{\bfgamma}_{EE})=o_p(1)$ and $\sum_{i=1}^n\big\{\hat{\bflambda}_{EE}^\T \bfg(\bfD_i,  \hat{\bfgamma}_{EE})\big\}^2=O_p(1)$ by \textit{Condition} 1, uniformly hold for $\lVert\hat{\bfgamma}_{EE}- \bfgamma_{0}\rVert\leq Cn^{-1/3}$, we have
	\begin{equation*}
	l=\hat{\bflambda}_{EE}^\T\sum_{i=1}^n\bfg(\bfD_i,  \hat{\bfgamma}_{EE})-\frac{1}{2} \hat{\bflambda}_{EE}^\T\Big\{ \sum_{i=1}^n\bfg(\bfD_i,  \hat{\bfgamma}_{EE}) \bfg^\T(\bfD_i,  \hat{\bfgamma}_{EE})\Big\} \hat{\bflambda}_{EE}+\red{o_p(1)}.
	\end{equation*}
	Furthermore, by applying the Taylor expansion for $\red{n^{-1/2}}\bfg(\bfD_i, \hat{\bfgamma}_{EE})$, we get
	\begin{equation*}
	\begin{split}
	&\hat{\bflambda}_{EE}^\T\sum_{i=1}^n\bfg(\bfD_i,  \hat{\bfgamma}_{EE})\\
	=&\red{n^{\frac{1}{2}}}\hat{\bflambda}_{EE}^\T\Big\{\red{n^{-\frac{1}{2}}}\sum_{i=1}^n\bfg(\bfD_i,  \bfgamma_{0})+ \red{n^{-\frac{1}{2}}} \sum_{i=1}^n \red{\frac{\partial\bfg(\bfD_i,  \bfgamma_0)}{\partial \bfgamma^\T}}(\hat{\bfgamma}_{EE}-\bfgamma_{0}) + o_p(\bf1) \Big\}\\
	=& \big(\red{n^{\frac{1}{2}}} \hat{\bflambda}_{EE}^\T\big) \bfS_n \big(\red{n^{\frac{1}{2}}} \hat{\bflambda}_{EE}\big)+o_p(1).
	\end{split}
	\end{equation*}
	The last equation is derived from (\ref{thm21}). Therefore, by \textit{Condition} 1 again, we have $l=(1/2)(\red{n^{\frac{1}{2}}}\hat{\bflambda}_{EE}^\T)\bfSigma_{11} (\red{n^{\frac{1}{2}}}\hat{\bflambda}_{EE})+o_p(1)$. Substituting $\hat{\bflambda}_{EE}$ by (\ref{thm22}), $l=(1/2)(\red{n^{\frac{1}{2}}}\bfQ_n^\T)\bfSigma^{-1}_{11} (\red{n^{\frac{1}{2}}}\bfQ_n)+o_p(1)$ with $\bfQ_n= (1/n) \sum_{i=1}^n \bfg(\bfD_i,  {\bfgamma}_{0}) + \bfSigma_{12}(\hat{\bfgamma}_{EE} -\bfgamma_{0})$, which completes the proof of Theorem 2.
	\end{proof}
	
	\subsection{Proof of Corollary 1}
	\begin{proof}
	
	Let us rewrite $\bfQ_n$ in the following manner:
	\begin{equation*}
	\begin{split}
	    \bfQ_n=&\frac{1}{n} \sum_{i=1}^n \bfg(\bfD_i,  \bfgamma_0)-\bfSigma_{12} \Big(E\frac{\partial\bfg_1}{\partial \bfgamma^\T}\Big)^{-1}\frac{1}{n} \sum_{i=1}^n \bfg_1(\bfD_i,  \red{\bfgamma_0})+\red{o_p(\bfn^{-\frac{1}{2}})}\\
	=&\bfSigma_{\ast}\frac{1}{n} \sum_{i=1}^n \bfg(\bfD_i,  \bfgamma_0)+\red{o_p(\bfn^{-\frac{1}{2}})},
	\end{split}
	\end{equation*}
	with $\bfSigma_{\ast}=\bfI_{L\times L}-\Big(\bfSigma_{12}\{E(\partial\bfg_1/\partial \bfgamma^\T)\}^{-1},\0_{L\times(L-p)}\Big)$. Here $\bfI_{L\times L}$ is an identity matrix. Accordingly, we rewrite $l$ as
	\begin{equation*}
	\begin{split}
	   l=&\frac{1}{2} \bigg\{ \red{n^{-\frac{1}{2}}}\sum_{i=1}^n \bfSigma_{11}^{-\frac{1}{2}} \bfg(\bfD_i,  \bfgamma_0) \bigg\}^\T \bfSigma_{11}^{\frac{1}{2}} \bfSigma_{\ast}^\T \bfSigma^{-1}_{11} \bfSigma_{\ast} \bfSigma_{11}^{\frac{1}{2}} \bigg\{ \red{n^{-\frac{1}{2}}}\sum_{i=1}^n \bfSigma_{11}^{-\frac{1}{2}} \bfg(\bfD_i,  \bfgamma_0) \bigg\}+o_p(1).
	\end{split}
	\end{equation*}
	
	Furthermore, according to the square-root decomposition, we have \red{$\bfOmega=\bfSigma_{11}^{1/2} \bfSigma_{\ast}^\T \bfSigma^{-1}_{11}\bfSigma_{\ast} \bfSigma_{11}^{1/2}=\bfP\bfLambda \bfP^\T$} with an orthogonal matrix $\bfP_{L\times\tilde{L}}$ and a diagonal matrix $\bfLambda_{L\times\tilde{L}}$ containing positive eigenvalues $\Lambda_1$,\ldots,$\Lambda_{\tilde{L}}$ of $\bfOmega$. Here $\tilde{L}$ represents the rank of the matrix $\bfOmega$ and $L$ is the length of $\bfg(\bfD, \bfgamma)$. Therefore, $l$ can be expressed as 
	\begin{equation*}
	l=\frac{1}{2} \sum_{j=1}^{\tilde{L}} \Lambda_j\Big[\red{n^{-\frac{1}{2}}}\sum_{i=1}^n\bfSigma_{11}^{-\frac{1}{2}} \bfg(\bfD_i,  \bfgamma_0)\Big]_j^2\red{+o_p(1)},
	\end{equation*}
	where $[\bfJ]_j$ indicates the $j^{th}$ element in vector $\bfJ$. Together with the fact of $\big[\red{n^{-1/2}}\sum_{i=1}^n \bfSigma_{11}^{-1/2}\bfg(\bfD_i, \bfgamma_0)\big]_j$ asymptotically followed by an independent standard normal distribution for $j=1,\ldots,\tilde{L}$, we conclude that $2l$ converges in distribution to $\sum_{j=1}^{\tilde{L}} \Lambda_j \chi_1^2 $, which is a weighted sum of standard $\chi^2$ distributions.

\end{proof}

\subsection{Proof of Theorem 3}
\begin{proof}
	For any $\bfgamma$ in the neighbourhood of $\bfgamma_*\neq\bfgamma_0$, we define $\tilde{\bflambda}(\bfgamma)=n^{-c}(\log n) \bar{\bfg}_n$, with $\bar{\bfg}_n=(1/n) \sum_{i=1}^n \bfg(\bfD_i,  \bfgamma)$ and $1/2<c<1$.
	
	First by Markov inequality and \textit{Condition} 3 \red{with some $\delta>0$}, we have \red{
	\begin{equation*}
	\sum_{i=1}^{\infty} P\Big(\lVert \bfg(\bfD_i,  \bfgamma) \rVert^{2} > i \Big) \leq \sum_{i=1}^{\infty}  \frac{E\lVert \bfg(\bfD_i,  \bfgamma) \rVert^{2+\delta}}{i^{1+\delta/2}}< \infty.
	\end{equation*}}
Applying the Borel-Cantelli Lemma, \red{we conclude that we can always find a large enough $N$ such that for any $i>N$, we have$\lVert \bfg(\bfD_i,  \bfgamma) \rVert \leq i^{-1/2}$ holds with probability one, which further implies $\max_{1\leq i\leq n} \lVert \bfg(\bfD_i,  \bfgamma) \rVert=o_p(n^{1/2})$.} Thus,
	\begin{equation*}
	\begin{split}
	&\max_{1\leq i\leq n} \lVert\tilde{\bflambda}^\T(\bfgamma) \bfg(\bfD_i,  \bfgamma)\rVert\\
	\leq& \lVert \tilde{\bflambda}(\bfgamma)\rVert\max_{1\leq i\leq n}\lVert\bfg(\bfD_i,  \bfgamma)\rVert
	= n^{\frac{1}{2}-c} \log(n) \lVert\bar{\bfg}_n\rVert =o_p(1),  
	\end{split}
	\end{equation*}
	where the the first inequality holds by Cauchy-Schwartz inequality and the last equality holds by \textit{Condition} 4. Therefore, with probability approaching to $1$, we have for all $1\leq i\leq n$, $1+\tilde{\bflambda}^\T(\bfgamma) \bfg(\bfD_i, \bfgamma)>0$. Finally,
	\begin{equation*}
	\begin{split}
	l&=\sup_{\bflambda} \sum_{i=1}^n \log\Big\{1+\bflambda^\T  \bfg(\bfD_i,  \bfgamma) \Big\}
	\geq \sum_{i=1}^n \log\Big\{1+\tilde{\bflambda}^\T(\bfgamma)  \bfg(\bfD_i,  \bfgamma) \Big\}\\
	&= \sum_{i=1}^n \tilde{\bflambda}^\T(\bfgamma)  \bfg(\bfD_i,  \bfgamma)+o_p(1)
	=n^{1-c} \lVert\bar{\bfg}_n\rVert^2\log(n)+o_p(1), 
	\end{split}
	\end{equation*}
	where the first equality holds by the property of the dual problem, the second equality holds by the first-order Taylor expansion of the function $\log (1+x)$ at $0$, and the final result holds under \textit{Condition} 4.
\end{proof}

\subsection{ Proof of Theorem 4}
\begin{proof}
	Given the fixed number of parameters in the full model, and $p$ and $p_0$ as the cardinalities of \red{candidate model $M$ and the true model $M_0$}, respectively, we show Theorem 4 in the following manner, which can be easily proved by the results from Theorems 2 and 3.
	$\textbf{(I)}~$$\textsc{ELCIC}(M)-\textsc{ELCIC}(M_0)>0$ with probability tending to $1$ for $M_0\nsubseteq M$.\\
	$\textbf{(II)}~$$\textsc{ELCIC}(M)-\textsc{ELCIC}(M_0)>0$ with probability tending to $1$ for $M_0\subseteq M$ and $p_0<p$. \\
    First, for $M_0\nsubseteq M$, applying Theorem 2 to $\textsc{ELCIC}(M_0)$ and Theorem 3 to $\textsc{ELCIC}(M)$, we derive
	\begin{equation*}
	\begin{split}
	\textsc{ELCIC}(M)-\textsc{ELCIC}(M_0)&=2l(M)+p\log(n)-\big\{2l(M_0)+p_0\log(n)\big\}\\
	&=2n^{1-c }\lVert\bar{\bfg}_n\rVert^2\log(n) -(\red{n^{\frac{1}{2}}}\bfQ_{n0}^\T)\bfSigma_{110}^{-1}(\red{n^{\frac{1}{2}}}\bfQ_{n0})\\ &+(p-p_0)\log(n)+o_p(1), 
	\end{split}
	\end{equation*}
	where $\bfQ_{n0}$ and $\bfSigma_{110}$ denote $\bfQ_{n}$ and $\bfSigma_{11}$ under the true model $M_0$. Notice that 
	\begin{equation*}
	P\Big[(\red{n^{\frac{1}{2}}}\bfQ_{n0}^\T)\bfSigma_{110}^{-1}(\red{n^{\frac{1}{2}}}\bfQ_{n0})\geq\log n\Big] \leq \frac{E\big(n\bfQ_{n0}^\T\bfSigma_{110}^{-1}\bfQ_{n0}\big)}{\log(n)}=\frac{tr(\bfSigma_{110}^{-1}\bfV)}{\log(n)}.
	\end{equation*}
	Applying \textit{Condition} 5, we have $(\red{n^{1/2}}\bfQ_{n0}^\T)\bfSigma_{110}^{-1} (\red{n^{1/2}}\bfQ_{n0}) =o_p(\log n)$, which further indicates that $n^{1-c}\lVert\bar{\bfg}_n\rVert^2\log(n) $ is the dominant term going to infinity under \textit{Condition} 4. Therefore, we have $\textsc{ELCIC}(M)-\textsc{ELCIC}(M_0)>0$ with probability tending to $1$  for $M_0\nsubseteq M$.
	
	Second, for $M_0\subseteq M$ and $p_0<p$, applying Theorem 2 to $\textsc{ELCIC}(M)$ and $\textsc{ELCIC}(M_0)$, we can derive
	\begin{equation*}
	\begin{split}
	\textsc{ELCIC}(M)-\textsc{ELCIC}(M_0)= &(\red{n^{\frac{1}{2}}}\bfQ_{n}^\T)\bfSigma_{11}^{-1} (\red{n^{\frac{1}{2}}}\bfQ_{n} )- (\red{n^{\frac{1}{2}}}\bfQ_{n0}^\T)\bfSigma_{110}^{-1}(\red{n^{\frac{1}{2}}}\bfQ_{n0})\\
	&+(p-p_0)\log(n)+o_p(1).
	\end{split}
	\end{equation*}
Since $(\red{n^{1/2}}\bfQ_{n0}^\T)\bfSigma_{110}^{-1} (\red{n^{1/2}}\bfQ_{n0})$ and $(\red{n^{1/2}}\bfQ_{n}^\T)\bfSigma_{11}^{-1} (\red{n^{1/2}}\bfQ_{n})$ have the same order $o_p(\log n)$ by the same argument above, we conclude that, for $M_0\subseteq M$ and $p_0<p$, $\lim\limits_{n\rightarrow\infty} P(\textsc{ELCIC}(M)-\textsc{ELCIC}(M_0)>0)=1$.
\end{proof}

\red{\section{Additional Simulation Studies}
\subsection{Variable selection in Case 2 under the GEE framework }
We consider the same setups in Case 2 in the main paper to only implement the variable selection by using the first part estimating equations in (3.11) in the main paper as our full estimating equations, and thus treating correlation coefficients as nuisance parameters. The selection rates by ELCIC and QIC are summarized in table \ref{tab1}, which further confirm the outperformance of ELCIC. 
\subsection{Variable selection for the augmented inverse probability weighting method}
To show unique and more general applications of our proposed criteria compared to the existing approaches, we provide an example for illustration, and under such context, the current existing criteria are not applicable. Here, we consider the augmented inverse probability weighted (AIPW) models with main focus on variable selection in the mean structure. Note the AIPW method has been popularly used to deal with missing data \citep{robins1994estimation} with extensive work in longitudinal data, survival analysis and causal inference \citep{bang2005doubly,seaman2009doubly,scharfstein1999adjusting,long2011} because of efficiency improvement and double robustness. For simplicity, here we only consider a simple linear regression with missing outcomes under the assumption of missing at random (MAR), but the extension to more complicated scenarios should be doable and straightforward. 

Suppose, for $i=1,\ldots,n$, we have the data where the outcomes $Y_i$ potentially missing, with $R_i$ as an observation indicator, i.e., $R_i=1$ if $Y_i$ is observed and $R_i=0$ otherwise. The covariates include $\bfX_i$ and $\bfS_i$. Also, we denote the observing probability of $Y_i$ as $\pi(\bfX_i,\bfS_i)=E(R_i\lvert\bfX_i,\bfS_i)$ parameterized by $\bfgamma$. The AIPW estimators are obtained by solving the following estimating equations
\begin{equation}\label{AIPW}
    \frac{1}{n}\sum_{i=1}^n\bigg\{\frac{R_i}{\hat{\pi}(\bfX_i, \bfS_i)}\bfU(Y_i,\bfX_i,\bfbeta)-\frac{R_i-\hat{\pi}(\bfX_i, \bfS_i)}{\hat{\pi}(\bfX_i, \bfS_i)}\tilde{\bfU}(\bfX_i,\bfS_i,\bfbeta)\bigg\}=\bf0
\end{equation}
where $\bfU(Y_i,\bfX_i,\bfbeta)=\bfX_i(Y_i-\mu_i(\bfbeta))$ and $\tilde{\bfU}(\bfX_i,\bfS_i, \bfbeta)=\bfX_i(\hat{a}(\bfX_i,\bfS_i)-\mu_i(\bfbeta))$, where $a(\cdot)$ is a function of $\bfX_i$ and $\bfS_i$ parameterized by $\bfalpha$ with $\hat{a}$ as some estimate of $E(Y_i\lvert\bfX_i,\bfS_i)$. 

As indicated in the literature, in addition to possible efficiency gains, one advantage of the AIPW estimator is that it is doubly robust, in the sense that it yields consistent results if either the missingness mechanism or the outcome regression model is correctly specified \citep{scharfstein1999adjusting}. However, it is challenging to correct specify $\hat{\pi}(\bfX_i,\bfS_i)$ or $\hat{a}(\bfX_i,\bfS_i)$ in practice due to limited prior knowledge, and thus the methods based on expectation of weighted quadratic mean square loss may not work \citep{shen2012, shen2017}. Also, likelihood based criteria are not applicable since semi-parametric approach is implemented here. In addition, the estimated quantities $\hat{\pi}(\bfX_i, \bfS_i)$ and $\hat{a}(\bfX_i,\bfS_i)$ make the model selection harder. However, our proposed ELCIC has great potential to deal with these issues, as we indicated in our method section, and the implementation is easy and straightforward. In particular, we take the formula on the left hand side of (\ref{AIPW}) as the full estimating equations in (2.5) in the main text, and also the nuisance parameters involved in $\hat{\pi}(\bfX_i, \bfS_i)$ and $\hat{a}(\bfX_i,\bfS_i)$ can be estimated from plug-in estimators. The consistency property still holds as long as these parameter estimates satisfy Condition 5 in the main text. 

We conduct extensive simulation studies to empirically evaluate the performance of ELCIC under the AIPW framework. Along with similar data structure and generation procedure in \cite{han2014multiply}, we first generate four mutually independent covariates $x_{1i}\sim$N$(5,1)$, $x_{2i}\sim$Bernoulli$(0.5)$, $x_{3i}\sim$N$(0,1)$,  $x_{4i}\sim$N$(0,1)$ and four auxiliary variables $s_{1i}=1+x_{1i}+2x_{2i}+\epsilon_{2i}$, $s_{2i}=I\{(s_{1i}+0.3\epsilon_{3i})>5.8\}$, $s_{3i}=\epsilon_{4i}$,$s_{4i}=x_2+\epsilon_{5i}$, where $\bfepsilon_i=(\epsilon_{1i},\ldots,\epsilon_{5i})^T$ follow a multivariate normal distribution with mean zeros and the covariance matrix $\bfSigma=$ with the diagonal elements valued by 1, the $(1,2)$ and $(2,1)$ entries as 0.5 and all others as $0$. The outcomes are generated from the linear model $Y_i=\beta_0+\beta_1 x_{1i}+\beta_2 x_{2i}+\beta_3 x_{3i}+\beta_4 x_{4i}+\epsilon_{1i}$ with $\bfbeta=(1,1,2,1,1)^{T}$. The true observing probability model is set to be logit$(\pi(\bfX_i,\bfS_i))=\gamma_0+\gamma_1 s_{1i}+\gamma_2 s_{2i}$ with $\bfgamma=(5, -1, 3)^{T}$, leading to the observing probability around $0.65$. We can also easily learn that the true imputation model should be $a(\bfX_i,\bfS_i)=\alpha_0+\alpha_ix_{1i}+\alpha_2x_{2i}+\alpha_3x_{3i}+\alpha_4x_{4i}+\alpha_5s_{1i}$. To further evaluate the effect of misspecification of either the model for missingness or the outcome regression model on our proposal's performance, we consider the following misspecified models: logit$(\pi^m(\bfX_i,\bfS_i))=\gamma^m_0+\gamma^m_1x_{1i}+\gamma^m_2x_{2i}+\gamma^m_3x_{3i}+\gamma^m_4x_{4i}+\gamma^m_5s_{1i}$ and
$a^m(\bfX_i,\bfS_i)=\alpha^m_0+\alpha^m_1s_{1i}+\alpha^m_2s_{2i}+\alpha^m_3s_{3i}$.

Thereafter, the variable selection is implemented based on our proposed ELCIC under four combinations evaluated: correct-specified $\pi(\bfX_i,\bfS_i)$ and correct-specified $a(\bfX_i,\bfS_i)$, correct-specified $\pi(\bfX_i,\bfS_i)$ and misspecified $a^m(\bfX_i,\bfS_i)$, misspecified $\pi^m(\bfX_i,\bfS_i)$ and correct-specified $a(\bfX_i,\bfS_i)$, denoted by PC$\_$IC, PC$\_$IM, and PM$\_$IC, respectively. Note that we do not consider the case with both misspecified $\pi^m(\bfX_i,\bfS_i)$ and $a^m(\bfX_i,\bfS_i)$ because the estimates will not be consistent,and the results based upon these inconsistent estimates are not reliable any more. For each scenario, we generate 500 Monte carlo data with sample size $n=250, 500$. Seven candidate models are considered for variable selection with selection rates recorded. The results are summarized in Table 2 in the Appendix. Overall, the selection rates for the correct mean structure are satisfactory with a high level (i.e., $>90\%$) when either of the models for missingness and outcome regression is correctly-specified, and increase as sample size becomes larger. In particular, when sample size $n=250$, the selection rate is up to $93.2\%$ when both models are correctly specified. For larger sample size (i.e., $n=500$), the results for PC$\_$IC, PC$\_$IM, PM$\_$IC are comparable, indicating our proposal is workable in the APIW framework. 

\subsection{Variable selection for the Ultra-high dimensional cases}
Per the suggestion by the reviewer and the editor, we conducted additional investigation on one potential extension of our proposed ELCIC criterion to ultra-high dimension variable selection situation via simulation studies. As we mentioned earlier, the proofs in theory for (ultra)high-dimensional cases with empirical likelihood are completely different with our current strategy, and are not trivial to be extended with more effort involved, which will be separately and fully pursued in our future studies. Here, we provide some empirical studies to show the challenges and numerical performance via simulation.

Note that there are two main issues for this extension. One issue is that the full model in ultra-high dimensional case could be substantially big, which may prevent the usage of classic empirical likelihood \citep{chen2009effects}. One tentative strategy is to consider the following modified criterion: 
\begin{equation}\label{ELCIC_ultra_1}
\begin{split}
       \text{ELCIC}^\ast=&-2\log R^R(\hat{\bflambda},\hat{\bfgamma})-n\sum_{j=1}^LP_{2,\nu}(\lvert\hat{\lambda}_j\lvert)\\
       &+n\sum_{k=1}^p P_{1,\pi}(\lvert\hat{\gamma}_k\lvert)+C_n\log(n)\text{df}_{\pi},
\end{split}
\end{equation}
where $\hat{\bflambda}$ and $\hat{\bfgamma}$ are solved by penalized empirical likelihood \citep{chang2018}. $P_{1,\pi}(\cdot)$ and $P_{2,\nu}(\cdot)$ are two penalty functions regulating the $\bfgamma$ and $\bflambda$ with tuning parameters $\pi$ and $\nu$, respectively. $C_n$ is a scaling factor diverging to infinitely at a slow rate \citep{tang2010}. $\text{df}_{\pi}$ is the number of nonzero coefficients in $\bfgamma$. Suggested by \cite{tang2010} and \cite{chang2018}, this ad-hoc criterion turns out to numerically work well but seems difficult to theoretically investigate its theoretical properties due to penalty functions with two extra tuning parameters involved. Another disadvantage is the complicated and unstable computational algorithm with high computation burden, which may lead to unsatisfactory results especially when generalized estimating equations embedded with empirical likelihood \citep{chang2018}. Therefore, it is not straightforward and feasible to apply into (\ref{ELCIC_ultra_1}) in ultra-high dimensional cases.

Here, we propose an alternative, a two-stage selection procedure, to resolve those issues but still make our proposed criterion easily implemented with satisfactory performance. In the first stage, we apply a non-parametric screening method, such as SIRS \citep{zhu2011model}, to reduce the ultra-high dimension into relatively low dimension. Then, in the second stage, we make use of the reduced model as the full model for ELCIC to capture the true one. Of note is that our proposed strategy has two advantages: 1) the first one is the consistency property because both SIRS procedure and ELCIC are proved to be consistent,thus making the combined two-stage selection consistent; 2) the second one is its easy and flexible implementation in practice. We will provide some numerical evaluation below to show its utility and advantage of ELCIC over other popularly used criteria.

We consider the following mean structure $\log(\mu_i)=\bfX_i^T\bfbeta$ for $i=1,\ldots,n$, where $\bfX_i$ is a covariate vector from a $p_n=1500-$dimensional multivariate normal distribution \textbf{MVN}(0,\bfV) with the covariance matrix $\bfV$ as an AR1 matrix with variance as 1 and correlation as 0.5. We generate outcomes from Poisson distribution, negative binomial distribution with its number of failure $k=2,4,8$, respectively. We first apply the SIRS to reduce the ultra-high dimension (i.e., $p_n=1500$) to relatively low dimension ($p_n=20$). Then, we apply the SCAD learning procedure to implement variable selection to the reduced candidate pool, where the tuning parameter is selected by cross-validation, BIC and ELCIC. Note that BIC is always set under the assumption of Poisson distribution, so that we can observe how mis-specification of distribution affects the performance of variable selection. 200 Monte Carlo data with sample size $n=250, 500$ are generated, and several important measurements are recorded, such as consistency, prediction error, false negative, false positive, and exact selection rate. The results are summarized in Table \ref{table3} in the Supplementary Material.

It is shown that when the outcomes are generated from Poisson distribution, BIC performs the best in terms of its low false positive and high exact selection rates. However, when the true distribution is negative binomial, BIC tends to select over-fitted model, and the performance becomes worse as the data have higher over-dispersion. To the contrast, ELCIC is observed to be robust of handling distribution mis-specification and have more stable selection performance. Cross-validation performs in the middle of ELCIC and BIC, but more time consuming. \chixiang{In addition, as sample size increases up to 750, ELCIC tends to have better selection rate and smaller false positives, while BIC has negligibly improved performance, and the false positives are not always reduced for the cross-validation approach.} Therefore, these numerical results justify the potential feasibility of applying ELCIC to ultra-high dimension, having much less computation burden and high stable performance in the meantime without sacrificing the distribution-free advantage. 
}

\bibliographystyle{chicago}
\bibliography{Bibliography-MM-MC}

\newpage
\begin{table}
	\centering
	\caption{Performance of ELCIC compared with QIC for the scenarios under longitudinal count data. 500 Monte Carlo datasets are generated with sample size $n=100, 300$ and the number of observations within-subject $T=3, 5$. The true mean structure is $\{x_1,x_2\}$ and an exchangeable (EXC) correlation structure with the correlation coefficient $\rho=0.5$ is the true model. Only variable selection is considered. }\label{tab1}
	\begin{tabular}{cccccccccc}
		\toprule
\multirow{2}{*}{Cluster Size} & \multirow{2}{*}{Sample Size} & \multirow{2}{*}{Criteria} & \multicolumn{6}{c}{Candidate Models}              \\
 &   &  & $x_1,x_2,x_3$ & \boldsymbol{$x_1,x_2$} & $x_1,x_3$ & $x_2,x_3$ & $x_1$ & $x_3$ \\
 \midrule
$T=3$ &	$n=100$  &	ELCIC	&	0.042	&	0.956	&	0	&	0.002	&	0	&	0	\\
&	&	QIC	&	0.14	&	0.86	&	0	&	0	&	0	&	0	\\
&	$n=300$  &	ELCIC	&	0.022	&	0.978	&	0	&	0	&	0	&	0	\\
&	&	QIC	&	0.116	&	0.884	&	0	&	0	&	0	&	0	\\ \midrule
$T=5$ &	$n=100$  &	ELCIC	&	0.034	&	0.966	&	0	&	0	&	0	&	0	\\
&	&	QIC	&	0.11	&	0.89	&	0	&	0	&	0	&	0	\\
&	$n=300$  &	ELCIC	&	0.018	&	0.982	&	0	&	0	&	0	&	0	\\
&	&	QIC	&	0.1	&	0.9	&	0	&	0	&	0	&	0	\\
		\bottomrule
	\end{tabular}
\end{table}

\begin{table}
\caption{\red{Performance of ELCIC for variable selection in the mean structure under the AIPW framework. 500 Monte Carlo data are generated with sample size $n=250, 500$. The model with $\{x_1,x_2,x_3,x_4\}$ is the true model. PC\_IC: correct-specified $\pi(\bfX_i,\bfS_i)$ and correct-specified $a(\bfX_i,\bfS_i)$; PC\_IM: correct-specified $\pi(\bfX_i,\bfS_i)$ and misspecified $a^m(\bfX_i,\bfS_i)$; PM\_IC: misspecified $\pi^m(\bfX_i,\bfS_i)$ and correct-specified $a(\bfX_i,\bfS_i)$.} }
\label{table2}
\begin{tabular}{ccccccccc}
\toprule
       &   & $x_1, x_2$ & $x_1, x_2, x_3$ & $x_1, x_2$ & \boldsymbol{$x_1, x_2$} & $x_1, x_2, x_3$ & $x_1, x_2, x_3$ & $x_1, x_2, x_3$ \\
          Model &$n$ &  &  & $x_3, s_3$ & \boldsymbol{$x_3, x_4$} & $x_4, s_3$ & $ x_4, s_4$ & $ x_4, s_3,s_4$ \\ \midrule
PC\_IC & 250& 0          & 0               & 0                    & 0.932                & 0.028                     & 0.040                      & 0                               \\
PC\_IC & 500 & 0          & 0               & 0                    & 0.950                 & 0.026                     & 0.024                     & 0                               \\
PC\_IM & 250& 0          & 0               & 0                    & 0.916                & 0.036                     & 0.046                     & 0.002                           \\
PC\_IM & 500& 0          & 0               & 0                    & 0.952                & 0.028                     & 0.020                      & 0                               \\
PM\_IC & 250 &0          & 0               & 0                    & 0.928                & 0.03                      & 0.042                     & 0                               \\
PM\_IC & 500 & 0          & 0               & 0                    & 0.958                & 0.022                     & 0.020                      & 0 
\\ \bottomrule 
\end{tabular}
\end{table}

\begin{table}
\caption{\chixiang{Performance of ELCIC compared with cross-validation (CV) and BIC (under Poisson distribution specification without over-dispersion) for two stage ultra-high variable selection in the mean structure under Poisson distribution with potential over-dispersed outcomes. 500 Monte Carlo data are generated with sample size $n=250, 500, 750$. NB: negative binomial; MS: consistency $\lVert\hat{\bfbeta}-\bfbeta_0\lVert^2$; PS: prediction $\lVert\bfX\hat{\bfbeta}-\bfX\bfbeta_0\lVert^2$; FN: false negative; FP: false positive; ES: exact selection rate; OS: over selection rate; US: under selection rate} }
\label{table3}
\footnotesize
\begin{tabular}{lcllllllcc}
\toprule
Scenario & $n$   & Criterion & MS   & PS   & FN   & FP    & ES   & OS   & US   \\ \midrule
POISSON  & $n=250$ & CV        & 0.22 & 0.15 & 0.34 & 1.32  & 0.52 & 0.24 & 0.10 \\
         &       & BIC       & 0.17 & 0.10 & 0.23 & 1.27  & 0.74 & 0.06 & 0.03 \\
         &       & ELCIC     & 0.17 & 0.11 & 0.27 & 1.05  & 0.59 & 0.20 & 0.12 \\
         & $n=500$ & CV        & 0.07 & 0.06 & 0.09 & 0.70  & 0.71 & 0.25 & 0.04 \\
         &       & BIC       & 0.04 & 0.02 & 0.02 & 0.09  & 0.96 & 0.03 & 0.01 \\
         &       & ELCIC     & 0.03 & 0.03 & 0.04 & 0.26  & 0.89 & 0.10 & 0.02 \\
         & \chixiang{$n=750$} & CV        & 0.05 & 0.04 & 0.04 & 0.68  & 0.72 & 0.27 & 0.02 \\
         &       & BIC       & 0.02 & 0.01 & 0.00 & 0.02  & 0.98 & 0.02 & 0.00 \\
         &       & ELCIC     & 0.02 & 0.01 & 0.01 & 0.13  & 0.92 & 0.08 & 0.01 \\ \midrule
NB $k=2$   & $n=250$ & CV        & 0.61 & 0.43 & 0.82 & 4.71  & 0.05 & 0.42 & 0.08 \\
         &       & BIC       & 0.63 & 0.46 & 0.57 & 9.35  & 0.00 & 0.54 & 0.01 \\
         &       & ELCIC     & 0.53 & 0.38 & 0.69 & 4.08  & 0.07 & 0.43 & 0.08 \\
         & $n=500$ & CV        & 0.19 & 0.14 & 0.13 & 4.06  & 0.17 & 0.74 & 0.05 \\
         &       & BIC       & 0.20 & 0.17 & 0.02 & 10.19 & 0.00 & 0.99 & 0.00 \\
         &       & ELCIC     & 0.15 & 0.12 & 0.07 & 2.68  & 0.27 & 0.69 & 0.04 \\
         & \chixiang{$n=750$} & CV        & 0.12 & 0.09 & 0.05 & 4.00  & 0.24 & 0.74 & 0.03 \\
         &       & BIC       & 0.15 & 0.13 & 0.00 & 10.94 & 0.00 & 1.00 & 0.00 \\
         &       & ELCIC     & 0.08 & 0.07 & 0.00 & 1.91  & 0.39 & 0.61 & 0.00 \\ \midrule
NB k=4   & $n=250$ & CV        & 0.32 & 0.20 & 0.45 & 3.44  & 0.13 & 0.52 & 0.05 \\
         &       & BIC       & 0.32 & 0.22 & 0.33 & 7.28  & 0.03 & 0.67 & 0.00 \\
         &       & ELCIC     & 0.28 & 0.18 & 0.37 & 2.66  & 0.23 & 0.46 & 0.07 \\
         & $n=500$ & CV        & 0.07 & 0.05 & 0.02 & 3.40  & 0.34 & 0.64 & 0.02 \\
         &       & BIC       & 0.09 & 0.08 & 0.00 & 7.98  & 0.01 & 1.00 & 0.00 \\
         &       & ELCIC     & 0.05 & 0.04 & 0.00 & 1.83  & 0.51 & 0.50 & 0.00 \\
         & \chixiang{$n=750$} & CV        & 0.05 & 0.04 & 0.00 & 3.05  & 0.47 & 0.54 & 0.00 \\
         &       & BIC       & 0.07 & 0.06 & 0.00 & 8.65  & 0.01 & 0.99 & 0.00 \\
         &       & ELCIC     & 0.04 & 0.03 & 0.00 & 1.34  & 0.59 & 0.42 & 0.00 \\\midrule
NB $k=6$   & $n=250$ & CV        & 0.27 & 0.17 & 0.40 & 3.36  & 0.24 & 0.45 & 0.06 \\
         &       & BIC       & 0.29 & 0.19 & 0.33 & 6.27  & 0.08 & 0.63 & 0.00 \\
         &       & ELCIC     & 0.26 & 0.16 & 0.37 & 2.55  & 0.28 & 0.42 & 0.09 \\
         & $n=500$ & CV        & 0.05 & 0.05 & 0.03 & 2.42  & 0.47 & 0.52 & 0.01 \\
         &       & BIC       & 0.06 & 0.05 & 0.01 & 6.43  & 0.08 & 0.92 & 0.00 \\
         &       & ELCIC     & 0.05 & 0.04 & 0.02 & 1.72  & 0.54 & 0.46 & 0.01 \\
         & \chixiang{$n=750$} & CV        & 0.03 & 0.03 & 0.00 & 2.99  & 0.49 & 0.51 & 0.00 \\
         &       & BIC       & 0.05 & 0.04 & 0.00 & 7.41  & 0.02 & 0.98 & 0.00 \\
         &       & ELCIC     & 0.03 & 0.02 & 0.00 & 1.33  & 0.61 & 0.39 & 0.00   \\ \bottomrule          
\end{tabular}
\end{table}